\documentclass[11pt]{amsart}
\usepackage{amsaddr}
\usepackage{amssymb}
\usepackage{latexsym}
\usepackage{cmll}
\usepackage{stmaryrd}
\usepackage{eqnarray}
\usepackage{mathtools}
\usepackage{xspace}

\usepackage{multirow}
\usepackage{lineno}

\newcommand{\ssection}{\section}

\newtheorem{lm}{Lemma}[section]
\newtheorem{thm}[lm]{Theorem}
\newtheorem{prp}[lm]{Proposition}
\newtheorem{cor}[lm]{Corollary}

\newtheorem{ex}[lm]{Example}
\newtheorem{df}[lm]{Definition}
\newtheorem{prob}{Problem}

\newcommand{\myqed}{\hfill$\Box$}

\newcommand{\wyj}[1]{\hfill{\small #1}}

\newcounter{senumi}[section]
\newcounter{senumip}[section]
\newcounter{temp}[section]

\def\thesenumi{\thesection.\arabic{senumip}}
\def\p@senumip\thesenumip{\thesenumi}
\newenvironment{senumerate}%
    {\begin{list}%
        {(\thesenumi)}%
        {\usecounter{senumip}}
        \setcounter{senumip}{\value{temp}}
    }%
    {\setcounter{temp}{\value{senumip}}
     \end{list}}
\newcounter{penumi}[section]

\newcounter{ptemp}[section]

   {\begin{enumerate}%
        \setcounter{penumi}{\value{ptemp}}%
        \setcounter{enumi}{\value{ptemp}}%
    }%
    {\setcounter{ptemp}{\value{enumi}}
     \end{enumerate}}
\newcounter{ppenumi}[section]

\newcounter{pptemp}[section]

\def\theppenumi{\theptemp.\arabic{ppenumi}}
    {\begin{list}%
        {(\theppenumi)}%
        {\usecounter{ppenumi}\setlength{\rightmargin}{\leftmargin}}
        \setcounter{ppenumi}{\value{pptemp}}
    }%
    {\setcounter{pptemp}{\value{ppenumi}}
     \end{list}}

\newcounter{entmp}

\newenvironment{tenumerate}{\begin{enumerate}}{\end{enumerate}}

\newcommand{\m}[1]{{\uppercase {\mathbf{#1}}}}
\newcommand{\rel}[1]{{\uppercase {\mathbb{#1}}}}
\newcommand{\mrel}[2]{{\uppercase{\mathbf{#1}}}\!\left[\uppercase{\mathbb{#2}}\right]}

\newcommand{\ceqv}[1]{\ensuremath{\operatorname{\textsc{Ceqv}
                                \ifthenelse{\equal{#1}{}}{}{\!\left( {\m #1} \right)}}}}
\newcommand{\csat}[1]{\ensuremath{\operatorname{\textsc{Csat}
                                \ifthenelse{\equal{#1}{}}{}{\!\left( {\m #1} \right)}}}}
\newcommand{\mcsat}[1]{\ensuremath{\operatorname{\textsc{MCsat}
                                \ifthenelse{\equal{#1}{}}{}{\!\left( {\m #1} \right)}}}}
\newcommand{\scsat}[1]{\ensuremath{\operatorname{\textsc{SCsat}
                                \ifthenelse{\equal{#1}{}}{}{\!\left( {\m #1} \right)}}}}
\newcommand{\csp}[1]{\ensuremath{\operatorname{\textsc{CSP}
                                \ifthenelse{\equal{#1}{}}{}{\!\left( {\rel #1} \right)}}}}

\newcommand{\polsatstar}{\csat}
\newcommand{\cpolsatstar}{\csat}

\newcommand{\npc}{\textsf{NP}-complete\xspace}
\newcommand{\conpc}{\textsf{co-NP}-complete\xspace}
\newcommand{\np}{\textsf{NP}\xspace}
\newcommand{\ptime}{\textsf{P}\xspace}
\newcommand{\usp}{Uniform Solution Property\xspace}
\newcommand{\USP}{USP\xspace}

\newcommand{\vr}[1]{{\uppercase {\mathcal {#1}}}}

\newcommand{\set}[1]{{\left\{ {#1} \right\} }}
\newcommand{\ci}{\subseteq}
\newcommand{\co}{\supseteq}

\newcommand{\card}[1]{\left| #1 \right|}
\newcommand{\cardd}[1]{\# #1}
\newcommand{\equa}[1]{\left\| #1 \right\|}
\newcommand{\tup}[1]{\langle #1 \rangle}

\newcommand{\intv}[2]{I\left[#1,#2\right]}
\newcommand{\vpair}[2]{{{#1}\choose{#2}}}

\renewcommand{\leq}{\leqslant}
\renewcommand{\geq}{\geqslant}
\renewcommand{\le}[1]{\leqslant_{#1}}

\newcommand{\comp}{\leq\!\geq}
\newcommand{\scomp}{<\!>}

\newcommand{\dist}[2]{{\sf dist}\!\left( #1,#2 \right)}
\newcommand{\distt}[3]{{\sf dist}_{#3}\!\left( #1,#2 \right)}

\renewcommand{\mapsto}{\longmapsto}
\newcommand{\tomaps}{\longmapsfrom}

\newcommand{\join}{\vee}
\newcommand{\meet}{\wedge}
\newcommand{\jjoin}{\bigvee}
\newcommand{\mmeet}{\bigwedge}
\newcommand {\bc}[1]{{\overline {#1}} }

\newcommand{\con}[1]{{\sf Con\:\m{#1}}}
\newcommand{\cn}[1]{{\sf Con\:\m{#1}}}
\newcommand{\Cn}[1]{{{\sf Con\:}\m{#1}}}
\newcommand{\cg}[3]{{\rm Cg}^{{\m {#1}}}({#2},{#3})}

\newcommand{\pol}[1]{{\rm Pol\:\m #1}}
\newcommand{\poln}[2]{{\rm Pol}_{#1}\m #2}
\newcommand{\clo}[1]{{\rm Clo\:\m #1}}

\renewcommand{\d}{\po D}
\newcommand{\q}{\po Q}

\newcommand{\po}[1]{{\mathbf {#1}}}
\newcommand{\te}[1]{{\mathbf {#1}}}
\newcommand{\tn} [1]{{\bf {#1}}}
\newcommand{\typ}{{\rm typ}}
\newcommand{\typset}[1]{\typ\set{#1}}
\newcommand{\rst}[2]{ {#1} |_{{#2}} }

\newcommand{\prect}[1]{\prec_{\tn #1}}

\newcommand{\minim}[3]{M_{\m #1}\left(#2,#3\right)}

\renewcommand{\o}[1]{\overline {#1}}
\newcounter{ttable}

\newcommand{\centr}[2]{\left(#2 : #1\right)}
\newcommand{\comm}[2]{\left[ #1 , #2 \right]}
\newcommand{\com}[2]{\left[ #1 , #2 \right]}
\newcommand{\commm}[2]{\left[ #1 , \ldots, #2 \right]}

\newcommand{\map}{\longrightarrow}

\newcommand{\congruent}[1]{\stackrel{#1}{\equiv}}

\newcommand{\h}[1]{\widehat{#1}}

\newcommand{\fj}{\varphi}

\begin{document}
\title{Satisfiability in multi-valued circuits}

\author{Pawe\l{} M. Idziak AND Jacek Krzaczkowski}

\email{idziak@tcs.uj.edu.pl, jacek.krzaczkowski@uj.edu.pl}
\address{Jagiellonian University,\\  Faculty of Mathematics and Computer Science,\\Department of Theoretical Computer Science,\\
ul. Prof. S. \L{}ojasiewicza 6,\\ 30-348, Krak\'ow, Poland}
\date{July 23, 2017}

\thanks{The project is partially supported by Polish NCN Grant \# 2014/14/A/ST6/00138.}

\maketitle

\begin{abstract}
Satisfiability of Boolean circuits is among the most known and important problems
in theoretical computer science.
This problem is \npc in general but becomes polynomial time when restricted either to monotone gates or linear gates.
We go outside Boolean realm and consider circuits built of any fixed set of gates
on an arbitrary large finite domain.
From the complexity point of view this is strictly connected with the problems of solving equations (or systems of equations) over finite algebras.

The research reported in this work was motivated by a desire to know
for which finite algebras $\m A$ there is a polynomial time algorithm
that decides if an equation over $\m A$ has a solution.
We are also looking for polynomial time algorithms that decide if two circuits over a finite algebra compute the same function.
Although we have not managed to solve these problems in the most general setting
we have obtained such a characterization for a very broad class of algebras
from congruence modular varieties.
This class includes most known and well-studied algebras
such as groups, rings, modules (and their generalizations like quasigroups, loops, near-rings, nonassociative rings, Lie algebras), lattices (and their extensions like Boolean algebras, Heyting algebras or other algebras connected with multi-valued logics including MV-algebras).

This paper seems to be the first systematic study of the computational complexity of satisfiability of non-Boolean circuits and solving equations over finite algebras.
The characterization results provided by the paper is given in terms of nice structural properties of algebras for which the problems are solvable in polynomial time. 
\end{abstract}

\ssection{Introduction
\label{sect-intro}}

One of the most celebrated \npc  problem is \textsc{SAT} --
the problem that takes on a Boolean expression 
and decides whether there is a $\set{0,1}$-valuation of variables
that satisfies this expression.

The most popular variant of this problem is \textsc{CNF-SAT}
(often called \textsc{SAT} as well)
in which the input is in Conjunctive Normal Form.
A formula in \textsc{CNF} is a conjunction of clauses each of which is a disjunction of
(e.g. at most 3) literals.
These clauses (if ternary) can be treated as (ternary) relations on the set $\set{0,1}$
and the SAT problem simply asks whether a conjunction of atomic formulas
(in this new relational language) is satisfiable.
This generalizes to any (finite) relational structure, say $\rel D$,
where the problem lies in answering whether a conjunction of atomic formulas
(in the language of $\rel D$) is satisfiable in $\rel D$.
This is now known under the name of Constraint Satisfaction Problem, or \csp{} for short.
A characterization of relational structures over $\set{0,1}$ for which \csp{}
is solvable in a polynomial time has been done in \cite{schaefer}.
The structures for which a polynomial time algorithm is not provided in \cite{schaefer}
have been shown there to be \npc with respect to \csp{}.
The similar dichotomy conjecture for \csp{} over arbitrary finite domains
has been stated by Feder and Vardi in \cite{feder-vardi}.
With the help of deep algebraic tools two algorithmic paradigms have been shown
to be fruitful in establishing polynomial time complexity of a wide range of relational structures.
One of these paradigms generalizes Gaussian elimination method
to the realm of algebras with few subpowers \cite{idziak-etal}.
The other generalizes DATALOG programming to local consistency checking method
\cite{barto-kozik:datalog}.
Both of those methods were explored to their limits,
so that they cannot be extended any further and a new 
approach is needed.
Very recently three independent proofs
(one by D.~Zhuk, another one by A.~Rafiey, J.~Kinne and T.~Feder
and the third one by A.~Bulatov)
of the \csp{} dichotomy conjecture have been announced.

In contrast to \textsc{CNF-SAT} the problem o satisfiability of general Boolean expression
is often called \textsc{CIRCUITS SAT} or \csat{} for short.
After restricting this \npc problem for example to the circuits that are
either monotone (only AND and OR gates)
or linear (only XOR gates)
the problem becomes solvable in a polynomial time.
Thus it is natural to isolate those collections of 2-valued gates that lead
to circuits with polynomially solvable satisfiability problem.
Actually such characterization of tractable families of 2-valued gates
can be inferred from the results of \cite{gorazd-krzacz:2el}.

In general, different collections of admissible gates (on a given set) give rise to algebras
(in the universal algebraic sense).
Thus we will talk about circuits over a fixed algebra $\m A$.
In this language the output gates of such circuits
can be represented by terms of an algebra $\m A$
(or polynomials of $\m A$, if values on some input gates are fixed).
We also relax the notion of satisfiability of such circuits to be read:

\begin{itemize}
  \item[]$\csat A$\\
        given a circuit over $\m A$ with two output gates $\po g_1, \po g_2$
        is there a valuation of input gates $\o x$ that gives the same output on
        $\po g_1, \po g_2$, i.e. $\po g_1(\o x)=\po g_2(\o x)$.
\end{itemize}
Note here, that in some cases (including $2$-element Boolean algebra)
the satisfiability of $\po g_1(\o x)=\po g_2(\o x)$
can be replaced by satisfiability of $\po g(\o x) = c$,
where $c$ is a constant and $\po g$ is a new output gate that combines $\po g_1$ and $\po g_2$.

In a circuit that has more than two output gates it is also natural to state the following question.
We will see that this very similar question has different taste.
\begin{itemize}
  \item[]$\mcsat A$\\
        given a circuit over $\m A$ with output gates $\po g_1, \po g_2, \ldots, \po g_k$
        is there a valuation of input gates $\o x$ that gives the same output on
        $\po g_1, \po g_2\ldots, \po g_k$,
        i.e. $\po g_1(\o x)=\po g_2(\o x)= \ldots = \po g_k(\o x)$.
\end{itemize}

From algebraic point of view problem $\csat A$ asks for the solutions
of an equation over $\m A$.
The problem $\mcsat A$ asks for solutions of a special system of equations over $\m A$.
But we can also ask for solutions of arbitrary systems of equations.
This however has a more natural wording in purely algebraic terms.
\begin{itemize}
  \item[]$\scsat A$\\
        Given polynomials
        \[
        \po g_1(\o x),\po h_1(\o x), \ \ldots \ , \po g_k(\o x),\po h_k(\o x)
        \]
        of an algebra $\m A$, is there a valuation of the variables $\o x$ in $A$ such that
        \begin{eqnarray*}
        \po g_1(x_1,\ldots,x_n) &=& \po h_1(x_1,\ldots,x_n)\\
                                &\vdots&\\
        \po g_k(x_1,\ldots,x_n) &=& \po h_k(x_1,\ldots,x_n),
\end{eqnarray*}
\end{itemize}

\medskip

With this natural approach via multi valued circuits also the problem \textsc{TAUTOLOGY}
has its natural generalization:
\begin{itemize}
  \item[]$\ceqv A$\\
        given a circuit over $\m A$
        is it true that for all inputs $\o x$
        we have the same values on given two output gates
        $\po g_1, \po g_2$, i.e. $\po g_1(\o x)=\po g_2(\o x)$.
\end{itemize}
In the algebraic setting this is simply the question of equivalence of two terms or polynomials.
Here equivalence of $k$ pairs of terms/polynomials reduces to $k$ independent
$\ceqv{}$ queries.

In Boolean realm the problem $\ceqv{}$ can be treated as the complement of $\csat{}$
and therefore is \conpc.
In general the closely related problem $\ceqv A$ is somehow independent from $\csat A$.
This independence means that all four possibilities of tractability/intractability can be witnessed by some finite algebras.
For example for the $2$-element lattice $\m L$ the problem $\csat L$ is in $\ptime$
while $\ceqv L$ is \conpc.
An example of a finite semigroup $\m S$ with $\ceqv S \in \ptime$ and $\csat S$ being \npc
can be inferred from \cite{klima:monoids}.

\medskip
It is worth to note that solving equations (or systems of equations) is one of the oldest
and well known mathematical problems
which for centuries was the driving force of research in algebra.
Let us only mention Galois theory, Gaussian elimination or Diophantine Equations.

In the decision version of these problems one asks if an equation (or system of such equations)
expressed in the language of a fixed algebra $\m A$,
has a solution in $\m A$.
In fact, for $\m A$ being the ring of integers this is the famous 10th Hilbert Problem on Diophantine Equations, which has been shown to be undecidable \cite{matiyasevich:10th}.
In finite realms such problems are obviously decidable in nondeterministic polynomial time.
There are numerous results related to problems connected with solving
equations and systems of equations over fixed finite algebras.
Most of them concerns well known algebraic structures as groups
\cite{burris-lawrence:groups}, \cite{goldman-russell},
\cite{horwath-etal}, \cite{horvath-szabo:polsatstar} rings
\cite{horvath:positive}, \cite{burris-lawrence:rings} or lattices
\cite{schwarz} but there are also some more general results
\cite{aichinger-mudrinski}, \cite{larose-zadori}.

\medskip

The main goal of this paper is to attack the classification problems of the form:
for which finite algebras $\m A$ there is an algorithm that answers one of the problems
$\csat A$, $\mcsat A$, $\scsat A$ or $\ceqv A$
in polynomial time with respect to the size of the circuit,
i.e. the size of the underlying graph of the circuit.
It seems that the most natural way to look at these problems is to treat circuits over $\m A$
(or in fact output gates of such circuits) as terms/polynomials of the algebra $\m A$.
This obvious translation makes our attack fruitful, as we can apply deep results and techniques developed by universal algebra such as {\em modular commutator theory} and
{\em tame congruence theory}.
These tools are especially useful in case of algebras generating congruence modular variety.
This assumption covers many well known structures as groups, rings, modules or lattices.
Our attempt to attack the classification problems has resulted
in partial characterization of computational
complexity of \csat{}, \mcsat{} and \ceqv{} for algebras generating congruence modular varieties.
This partial characterization leaves some room to be filled before establishing a dichotomy.

\ssection{The results
\label{sect-results}}

In this section we present the state of the art in more details
and discuss our results and tools.

The first thing in which our research differs from what has been already considered
is that we concentrate on circuits rather than on syntactic form of terms or polynomials.
This difference is visible in how the size of the input is measured.
We have seen how an output gate can be treated as a term or a polynomial.
On the other hand, in an obvious way,
every term over $\m A$ can be treated as a circuit
in which each gate is used as an input to at most one other gate.
This leads to a circuit whose underlying graph is a tree.
However circuits can have more compact representation than terms.
For example, in groups the terms
$\po t_n(x_1,x_2,\ldots,x_n) = [\ldots [[x_1,x_2],x_3] \ldots x_n]$, 
(where $[x,y]=x^{-1}y^{-1}xy$ is the group commutator)
expressed in the pure group language of $(\cdot,\ ^{-1})$ have an exponential size in $n$,
as the number of occurrences of variables doubles whenever we pass from $n$ to $n+1$.
On the other hand the size of a circuit realizing $\po t_n$ has $6n-5$ vertices
as can be seen from the picture below.

\bigskip
\begin{center}
\includegraphics[width=6cm,height=7.2cm]{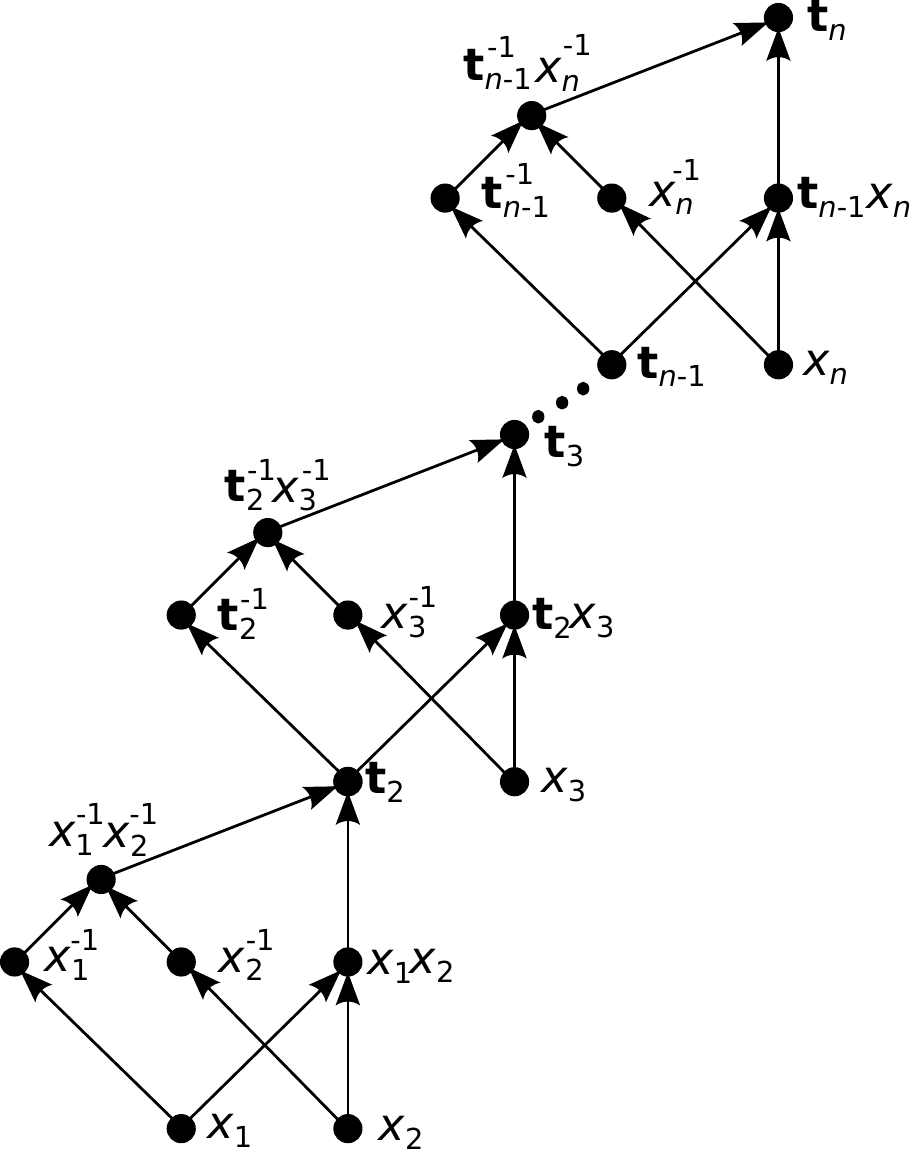}
\end{center}

\bigskip

The consequences of this (exponential) disproportion in measuring the input size for terms and circuits are illustrated by the following example.

\begin{ex}
\label{exam-s3}
There are finite groups $\m A$ such that
$\csat A$ is \npc, while there are polynomial time algorithms
for solving equations over $\m A$.

There are also finite groups $\m B$ such that
$\ceqv A$ is \conpc, while there are polynomial time algorithms
for checking the identities in $\m B$.
\end{ex}

\begin{proof}
The first such example for $\csat{}$ was the symmetric group $\m S_3$
for which polynomial time algorithm was shown in \cite{horvath-szabo:groups},
while the first author's observation on the \npc{}ness is included in \cite{gorazd-krzacz:preprimal}.

The papers \cite{horvath:metabelian,horvath-szabo:polsatstar,horvath-szabo:a4}
contain many other examples of solvable non-nilpotent groups
which witness both statements in our example.
\end{proof}

Note that in case of $\scsat{}$ there is no such disproportion in the size
as every polynomial equation $z = \po t(\o x)$
can be replaced by a system of equations of the form $y=\po f(x_1,\ldots,x_k)$ or $y=c$,
where $\po f$ is one of the basic operations and $c$ is a constant.
This replacement has linear size with respect to the circuit representing $\po t(\o x)$.
For example for the above term
$\po t_n(x_1,x_2,\ldots,x_n) = [\ldots [[x_1,x_2],x_3] \ldots x_n]$,
slightly abusing our conditions,
we can use the following representation
\begin{eqnarray*}
t_2&=& x^{-1}_1 x^{-1}_2 x_1 x_2\\
t_3&=& t^{-1}_2 x^{-1}_3 t_2 x_3 \\
   &\vdots&  \\
t_n&=& t^{-1}_{n-1} x^{-1}_n t_{n-1} x_n,
\end{eqnarray*}
in which $t_2,\ldots,t_n$ are treated as variables.

However, even in the setting of a single equation,
representing a polynomial $\po t(\o x)$ by its corresponding circuit
and looking at the size of this circuit (instead of the syntactic length of $\po t$)
allows us to harmlessly expand the original language of the algebra $\m A$ by finitely many polynomials.
In fact in our intractability proofs we will often expand the language
of the original algebra $\m A$ by finitely many polynomials of $\m A$.
This will allow us to code \npc problems in much more smooth way.
Note that the possibility of such expansions show that the characterizations we are looking for
can be done up to polynomial equivalence of algebras.
Two algebras are polynomially equivalent if they have the same universes and each polynomial of one of them can be defined by composing the polynomials of the other one.

It turns out that quite a few results on the complexity
of the problems $\csat{}$, $\mcsat{}$, $\scsat{}$ and $\ceqv{}$
are already known for particular kinds of (finite) algebras.

\medskip
\noindent

\begin{ex}
\label{ex-groups}
Finite Groups:
\begin{itemize}
  \item If $\m A$ is Abelian then $\scsat A \in \ptime$ (by Gaussian elimination),
        and for all other groups $\scsat A$ is \npc \
        \cite{goldman-russell}.
  \item $\csat A$ is in $\ptime$, whenever $\m A$ is nilpotent \cite{goldman-russell}
        and \npc \ otherwise \cite{goldman-russell,horvath-szabo:polsatstar}.
  \item $\ceqv A$ is in $\ptime$, whenever $\m A$ is nilpotent \cite{burris-lawrence:groups}
        and \conpc \ otherwise \cite{horwath-etal,horvath-szabo:polsatstar}.
\end{itemize}
\end{ex}

\begin{ex}
\label{ex-rings}
Finite Rings:
\begin{itemize}
  \item If $\m A$ is essentially an Abelian group (i.e. satisfies the identity $xy=0$)
        then $\scsat A \in \ptime$ (by Gaussian elimination),
        and for all other rings $\scsat A$ is \npc \ \cite{larose-zadori}.
  \item $\csat A$ is in $\ptime$, whenever $\m A$ is nilpotent \cite{horvath:positive}
        and \npc \ otherwise \cite{burris-lawrence:rings}.
  \item $\ceqv A$ is in $\ptime$, whenever $\m A$ is nilpotent
        and \npc \ otherwise (see \cite{hunt-stearns} for commutative rings and          \cite{burris-lawrence:rings} for general case).
\end{itemize}
\end{ex}

\begin{ex}
\label{ex-lattices}
\medskip
\noindent
Finite Lattices:
\begin{itemize}
  \item $\csat A \in \ptime$ if $\m A$ is distributive
        and \npc \ otherwise \cite{schwarz}.
  \item For all nontrivial lattices $\m A$, $\scsat A$ is \npc \
        while $\ceqv A$ is \conpc \ (easy to see).
\end{itemize}
\end{ex}

The examples given above suggest that the existence of polynomial time algorithms for the considered circuits problems go hand in hand with nice structure theory of the underlying algebras.
However there are only two results that can be considered general enough to be expressed in structural terms.
These results are stated in the following two theorems.

First note that E.~Aichinger and N.~Mudrinski \cite{aichinger-mudrinski} have shown the following
theorem, a partial converse of which is our Theorem \ref{int-ceqv-nonnil}.

\begin{thm}
\label{int-ceqv-supernil}
If $\m A$ is a finite supernilpotent algebra from a congruence variety
then $\ceqv A$  is in $\ptime$.
\end{thm}

The second general result is that of B.~Larose and L.~Z{\'a}dori \cite{larose-zadori}.
After observing that $\scsat{}$ has exactly the same expressive power as $\csp{}$
they used mutual translation between $\scsat{}$ and $\csp{}$ to prove the first part of the next characterization, while the second one is a form of Gaussian elimination.
\begin{thm}
\label{int-scsat-zadori}
For a finite algebra $\m A$ from a congruence modular variety:
\begin{itemize}
  \item if $\scsat A$ is not \npc then  $\m A$ is affine
        (i.e. $\m A$ is polynomially equivalent to a module over a finite ring),
  \item if $\m A$ is affine then $\scsat A \in \ptime$.
\end{itemize}
\end{thm}
Not as much is known when one leaves the congruence modularity realm.
It is worth to note however that an important extension of Theorem \ref{int-scsat-zadori}
to finite algebras from varieties omitting $\tn 1$
(in the sense of \cite{hm}) can be found in \cite{zadori}.

Also a number of results on semigroups do not fall in congruence modular setting
but these results are still about particular type of algebras.
The paper \cite{klima-etal} gives a nice, but somewhat technical,
characterization of finite monoids $\m A$ for which $\scsat A \in P$.
There are also several results on the complexity of $\scsat A$
for  particular semigroups or classes of semigroups,
but we are far from having a full characterization similar to that for monoids.
This is because the paper \cite{klima-etal} contains a proof that the expressive power of
$\scsat A$ over semigroups is equivalent to the expressive power of CSP.
Surprisingly another class of algebras with the same expressive power
is the class of algebras with unary operations only \cite{broniek, feder-etal}.

In Section \ref{sect-ex} we will prove that the expressive power of $\csat{}$
is no weaker that this of CSP, as expressed below.

\begin{prp}
\label{prp-csp-csat}
For every finite relational structure $\rel D$ (with finitely many relations)
there is a finite algebra $\mrel A D$
such that the problem $\csp D$ is polynomially equivalent to $\csat {A[\rel D]}$.
\end{prp}

Unlike in the $\scsat{}$ setting we do not know whether the expressive power of $\csat{}$ is not bigger than the one of $\csp{}$.

\begin{prob}
\label{prob-csat-csp}
Is it true that for every finite algebra $\m A$ there exists a relational structure
$\rel D[\m A]$ such that the problems
$\csat A$ and $\csp {{\rel D}[\m A]}$ are polynomially equivalent?
\end{prob}

The above difference between a single equation and a system of equations
is probably a consequence of the presence of an external conjunction
in systems of equations.
Intuitively, to replace a system of equations by a single equation,
one needs to squeeze many terms (or polynomials) into a single one.
This requires an analogue of an internal conjunction present in Boolean algebras.
Since such a squeeze is not always possible, more algebras
may have polynomial time algorithms for $\csat{}$ than for $\scsat{}$.
Actually our work is going to confirm this claim.

\bigskip

One of the main difficulties in characterizing finite algebras with $\scsat A \in \ptime$ is that this property does not carry over quotient algebras (unless $\ptime = \np$).
The paper \cite{klima-etal} contains an example of a finite semigroup $\m A$ and its congruence $\theta$ with $\scsat {A/\theta}$ being \npc \ while $\scsat A \in \ptime$.
The example below (an easy proof of which is postponed to Section \ref{sect-ex})
shows that this unwanted phenomena occurs
for the $\csat{}$ problem, as well.

\begin{ex}
\label{ex-non-hom}
There is a finite algebra $\m A$ and its congruence $\theta$ such that
$\csat A \in \ptime$ while $\csat {A/\theta}$ is \npc.
\end{ex}

Since passing to quotient algebras may not preserve polynomial time complexity for $\csat{}$,
it is natural to work under the stronger assumption that not only
$\csat A \in \ptime$, but $\csat {A/\theta} \in \ptime$ for all congruences $\theta$ of $\m A$.
Such assumption has also a natural interpretation.
Given $\m A$ we want a fast method to solve equations over $\m A$, or at least decide if such equations have solutions.
However such solutions may not exist in the original algebra $\m A$.
They obviously do exist in $\m A/1_{\m A}$,
where $1_{\m A}$ is the congruence collapsing everything.
Thus the best we can do, is to determine (existence of)
the solutions with best possible precision, i.e. modulo the smallest congruences possible.
This however requires $\m A$ to be regular enough so that $\csat{A^\prime}$ is in \ptime for all quotients $\m A^\prime$ of $\m A$.

\medskip
After fixing the setting we are working in, we can state our main result in the next theorem
which in fact summarizes Theorems \ref{thm-nil-dl-decomp} and \ref{thm-supernil-dl-decomp}.

\begin{thm}
\label{int-csat}
Let $\m A$ be a finite algebra from a congruence modular variety.
\begin{enumerate}
\item
\label{int-csat-nonpc}
    If $\m A$ has no quotient $\m A^\prime$ with $\csat{A^\prime}$ being \npc
    then $\m A$ is isomorphic to a direct product $\m N \times \m D$,
    where $\m N$ is a nilpotent algebra
    and $\m D$ is a subdirect product of $2$-element algebras
    each of which is polynomially equivalent to the $2$-element lattice.
\item
\label{int-csat-ptime}
    If $\m A$ decomposes into a direct product $\m N \times \m D$,
    where $\m N$ is a supernilpotent algebra
    and $\m D$ is a subdirect product of $2$-element algebras each of which is polynomially equivalent to the $2$-element lattice,
    then for every quotient $\m A^\prime$ of $\m A$ the problem
    $\csat{A^\prime}$ is solvable in polynomial time.
\end{enumerate}
\end{thm}

To understand the above result first note that the congruence modularity assumption covers most algebraic structures considered in classical mathematics.
In particular it includes groups (and their extensions like rings, fields),
and lattices (and their extensions like Boolean algebras or other algebras connected with multi-valued logics).
This assumption does not cover however semigroups (or even semilattices) or multiunary algebras.

The conditions (\ref{int-csat-nonpc}) and (\ref{int-csat-ptime})
show that the nilpotent groups and rings as well as distributive lattices
mentioned in Examples \ref{ex-groups}, \ref{ex-rings} and \ref{ex-lattices}
are in fact paradigms for $\csat{}$ tractability in congruence modular realm.
In fact the structural conditions described in Theorem \ref{int-csat}, 
when specialized to groups, rings or lattices,
gives the already known characterizations presented in the Examples.

The decomposition enforced in (\ref{int-csat-nonpc}) is a result of
almost a dozen constructions interpreting  \npc problems
(mostly \textsc{SAT} and $k$-\textsc{Colorability})
into $\csat A$, whenever $\m A$, or some of its quotients,
fails to satisfy one of the structural conditions that finally lead to this nice decomposition.

The second factor, $\m D$,  of this decomposition is easier to understand than the first one.
It essentially behaves like a finite distributive lattice,
but the algebra $\m D$ does not need to actually have lattice operations.
Instead $\m D$ is composed of $2$-element algebras each of which does have lattice operations expressible by polynomials (and all other operations monotone with respect to this lattice order).

The first factor, $\m N$, of this decomposition requires
the general algebraic notion of nilpotency in congruence modular setting that goes back
to the late 1970's when Smith \cite{smith},
Hagemann and Herrmann \cite{hagemann-herrmann}, Gumm \cite{gumm}
and finally Freese and McKenzie \cite{fm} developed necessary deep tools
of {\em modular commutator theory}.
In fact a notion of the commutator multiplication $\comm{\alpha}{\beta}$
of congruences $\alpha, \beta$ of arbitrary algebras
was defined in a way
that extends multiplication of ideals in ring theory
and commutator multiplication of normal subgroups in group theory.
With the help of such commutator one can define solvable and nilpotent congruences and algebras.

Finite nilpotent groups (and rings) behave very nicely.
In particular they decompose into direct products of groups (or rings) of prime power order.
Unfortunately such nice decomposition of nilpotent algebras in congruence modular varieties
does not hold in general.
However, in this general setting, nilpotent algebras that have this nice decomposition
(and have only finitely many basic operations)
are exactly those that are supernilpotent.
In fact supernilpotency has been introduced by another universal algebraic generalization of commutator multiplication of congruences.

The nilpotent/supernilpotent gap that occurs in Theorem \ref{int-csat}
resists to be easily filled.
This is because in supernilpotent case there is a bound on the arity
of the so called commutator polynomials.
These commutator polynomials can imitate the behavior of the long conjunction.
In nilpotent (but not supernilpotent) case arbitrary long conjunctions are expressible.
But this can be probably done at the expense of exponentially large (with respect to the arity)
circuits needed to represent those conjunctions.
This expected exponential size probably prevents polynomial time reduction of \npc problems to $\csat{}$ in nilpotent but not supernilpotent case.

The reductions we have produced to show intractability of the considered problems
are based on the local behavior described by the second deep tool of universal algebra
known as {\em tame congruence theory}.
This theory, created and described by D.~Hobby and R.~McKenzie in \cite{hm},
is a perfect tool for studying the local structure of finite algebras.
Instead of considering the whole algebra and all of its operations at once,
tame congruence theory allows us to localize to small subsets
on which the structure is much simpler to understand and to handle.
According to this theory
there are only five possible ways a finite algebra can behave locally.
The local behavior must be one of the following:
\begin{enumerate}
\item[{\tn 1}.]  a finite set with a group action on it,
\item[{\tn 2}.]  a finite vector space over a finite field,
\item[{\tn 3}.]  a two element Boolean algebra,
\item[{\tn 4}.]  a two element lattice,
\item[{\tn 5}.]  a two element semilattice.
\end{enumerate}
Now, if from our point of view a local behavior of an algebra is `bad' then we can often show that
the algebra itself behaves `badly'.
For example, since $\csat{}$ or $\ceqv{}$ is intractable in $2$-element Boolean algebra
one can argue that in any finite algebra with tractable $\csat{}$ or $\ceqv{}$
type $\tn 3$ cannot occur (see Theorem \ref{thm-type-3}).

On the other hand it is not true that if the local behavior is `good'
then the global one is good as well.
Several kinds of interactions between these small sets
can produce a fairly messy global behavior.
Such interactions often contribute to \npc{}ness of the considered problems
(see for example Lemma \ref{lm-solvable}).
Also the relative `geographical layout' of those small sets
can result in unpredictable phenomena,
as in Theorems \ref{thm-24transfer} and \ref{thm-42transfer}.

\medskip

Combining Theorems \ref{int-csat} and \ref{int-scsat-zadori}
we are able to infer the following corollary.

\begin{cor}
\label{int-mcsat}
Let $\m A$ be a finite algebra from a congruence modular variety.
\begin{enumerate}
\item
\label{int-csat-nonpc}
    If $\m A$ has no quotient $\m A^\prime$ with $\mcsat{A^\prime}$ being \npc
    then $\m A$ is isomorphic to a direct product $\m M \times \m D$,
    where $\m M$ is an affine algebra
    and $\m D$ is a subdirect product of $2$-element algebras
    each of which is polynomially equivalent to the $2$-element lattice.
\item
\label{int-csat-ptime}
    If $\m A$ decomposes into a direct product $\m M \times \m D$,
    where $\m M$ is an affine algebra
    and $\m D$ is a subdirect product of $2$-element algebras each of which is polynomially equivalent to the $2$-element lattice,
    then for every quotient $\m A^\prime$ of $\m A$ the problem
    $\csat{A^\prime}$ is solvable in polynomial time.
\end{enumerate}
\end{cor}

Our constructions used to show that lack of nice structure of the algebra $\m A$
leads to intractability of $\csat A$  can be also modified to work for intractability of $\ceqv A$ so that we are able to prove a partial converse to Theorem \ref{int-ceqv-supernil}.

\begin{thm}
\label{int-ceqv-nonnil}
Let $\m A$ be a finite algebra from a congruence modular variety.
If $\m A$ has no quotient $\m A^\prime$ with $\ceqv{A^\prime}$ being \conpc
then $\m A$ is nilpotent.
\end{thm}

A short informal summary of these results is completed in the following table,
where `DL-like' stays for being a subdirect product of algebras polynomially equivalent to $2$-element lattices.

\bigskip\bigskip
\bigskip\bigskip
\begin{center}
\begin{tabular}{|l|l|l|l|}
\hline
                &tractable
                &\multicolumn{1}{c|}{\multirow{2}*{open}}
                &intractable
\\
                &(polynomial time)
                &
                &(\textsf{co-NP}- or \npc)
\\ \hline\hline
\multirow{2}*{$\ceqv{}$}
                &supernilpotent
                &nilpotent
                &non nilpotent
\\
                &\wyj{Aichinger {\em\&} Mudrinski \cite{aichinger-mudrinski}}
                &but not supernilpotent
                &\wyj{Thm \ref{int-ceqv-nonnil}}
\\ \hline\hline
\multirow{2}*{$\csat{}$}
                &supernilpotent $\times$ DL-like
                &nilpotent
                &non (nilpotent $\times$ DL-like)
\\
                &\wyj{Thm \ref{int-csat} (\ref{int-csat-ptime})}
                &but not supernilpotent
                &\wyj{Thm \ref{int-csat} (\ref{int-csat-nonpc})}
\\ \hline
\multirow{2}*{$\mcsat{}$}
                &affine $\times$ DL-like
                &\multicolumn{1}{c|}{\multirow{2}*{---}}
                &otherwise
\\
                &\wyj{Cor \ref{int-mcsat} (\ref{int-csat-ptime})}
                &
                &\wyj{Cor \ref{int-mcsat} (\ref{int-csat-nonpc})}
\\ \hline
\multirow{2}*{$\scsat{}$}
                &affine
                &\multicolumn{1}{c|}{\multirow{2}*{---}}
                &otherwise
\\
                &\wyj{Gaussian elimination}
                &
                &\wyj{Larose {\em\&} Z{\'a}dori \cite{larose-zadori}}
\\ \hline
\end{tabular}
\end{center}

\bigskip

An obvious open question is the following:

\begin{prob}
\label{prob-nil-supernil}
Determine the computational complexity of $\ceqv{}$ and $\csat{}$ for nilpotent, but not supernilpotent finite algebras from congruence modular varieties.
\end{prob}

Another question that arises naturally is the role of quotient algebras in the proofs of
\npc{}ness of considered problems.
Note that the result of B.~Larose and L.~Z{\'a}dori \cite{larose-zadori}
for $\scsat{}$ makes no use of quotient algebras.
This is because a quotient of an affine algebra is affine itself.

Example \ref{ex-non-hom} shows that in general it is not enough to establish
\npc{}ness for a quotient algebra to conclude it for the original one.
However it may suffice in some more restricted setting like for example congruence modularity.
In concrete algebraic structures where basic operations are described explicitly it might be much easier.  
In fact in structures described in Examples \ref{ex-groups}, \ref{ex-rings} and \ref{ex-lattices}, passing to quotients is hidden in the hardness proofs 
and (implicitly) replaced by an involved control over congruences in groups, rings or lattices, respectively.

\begin{prob}
\label{prob-modular-homo}
Is it true that \npc{}ness of $\csat{}$ for some quotient of a finite algebra $\m a$
from a congruence modular variety
implies \npc{}ness of $\csat{}$ for $\m A$ itself.
\end{prob}

Even if the answer to Problem \ref{prob-modular-homo} would be negative the next one remains open.

\begin{prob}
\label{prob-cm-without-homo}
Do the characterizations of Theorems \ref{int-csat} (\ref{int-csat-nonpc}), \ref{int-ceqv-nonnil}
and Corollary \ref{int-mcsat} (\ref{int-csat-nonpc})
remain true without passing to quotient algebras.
\end{prob}

Note here that when restricting to equations of the form
\[
\po t(x_1,\ldots,x_n) = c
\]
where $\po t$ is a polynomial but $c$ is a constant,
the satisfiability in the quotient $\m A/\theta$ reduces
to the satisfiability of at least one of the equations in the following disjunction
\[
\po t(x_1,\ldots,x_n) = c_1 \ \  \lor \ \ldots \ \lor  \ \ \po t(x_1,\ldots,x_n) = c_s,
\]
where $\set{c_1,\ldots,c_s}$ is the equivalence class of $c$ modulo $\theta$.
This Cook style reduction gives the hope to attack the following problem.

\begin{prob}
Characterize finite algebras $\m A$ for which determining the existence of a solution to the equations of the form $\po t(\o x) = c$ can be done in polynomial time.
\end{prob}

In view of Problem \ref{prob-csat-csp},
the natural conjecture about dichotomy for $\csat{}$
is not so evident.
However there is a slightly bigger hope for such dichotomy after:
\begin{itemize}
  \item restricting $\csat{}$ to the equations of the form $\po t(\o x) = c$, and
  \item relaxing many-to-one reductions to Cook reductions.
\end{itemize}

\begin{prob}
Prove the dichotomy in the above settings.
\end{prob}

\ssection{Background material
\label{sect-background}}

In general we use the terminology and notation of \cite{mmt}.
Our brief introduction to this terminology, notation and the facts
that we are using in this paper
is modelled after that in \cite{berman-idziak}.

An {\em algebra} $\m A =\langle A, \te f_i (i\in I)\rangle$ is a
nonvoid set $A$ together with a collection of finitary operations
$\te f_i$ on $A$ indexed by a set $I$.
The set $A$ is called the {\em universe} of the algebra and the $\te f_i$'s
are the {\em fundamental operations} of $\m A$.
For $i \in I$ the operation $\te f_i$ maps
$A^{n_i}$ to $A$, that is, $\te f_i$ is $n_i$-ary.
The function from $I$ to the integers given by $i\mapsto n_i$ is the
{\em similarity type} of the algebra $\m A$.
If $I$ is finite, then the algebra is said to be of {\em finite similarity type}.
For algebras of  finite similarity type we often just list the operations,
e.g., a Boolean algebra might be given as $\m B = \langle B,\land,\lor,\neg,0,1\rangle$.
An algebra is {\em finite} if its
universe is finite and is {\em trivial} if its universe has only
one element.

An algebra $\m A =\langle A, \te f_i (i\in I)\rangle$ may also be viewed
as a model in the language $L$ where $L$ consists
of all the  function symbols $\te f_i$ for  $i\in I$.
When necessary, we distinguish  the function
symbol $\te f_i$ in $L$ from the fundamental operation $\te f_i$
on $A$ by writing $\te f_i^{\m A}$
to denote the $n_i$-ary operation on the algebra $\m A$.
A {\em term} for $L$ over a set of variables $X=\set{x_1,x_2,\dots}$
is defined inductively by letting every $x_j\in X$ be a term and
if $i\in I$ and $\te t_1,\dots,\te t_{n_i}$ are terms, then
$\te f_i(\te t_1,\dots,\te t_{n_i})$ is also a term.
If the variables that appear in a term $\te t$ are in the set
$\set{x_1\dots,x_n}$,
then we say $\te t$ is $n$-ary  and denote this by writing $\te t(x_1,\dots,x_n)$.
If $\te t(x_1,\dots,x_n)$ is an $n$-ary term for $L$ over
$X$ and $\m A$ is an algebra in the language $L$, then the
{\em term operation}
$\te t^{\m A}$ on $\m A$ corresponding to $\te t$ is defined by letting
$x_i^{\m A}$ be the projection on the $i$-th coordinate and if
\[
\te t(x_1,\dots,x_n) =
\te f_i(\te t_1(x_1,\dots,x_n),\dots,\te t_{n_i}(x_1,\dots,x_n)),
\]
then
\[
\te t^{\m A}(a_1,\dots,a_n)=
\te f_i^{\m A}(\te t_1^{\m A}(a_1,\dots,a_n),\dots,\te t_{n_i}^{\m A}(a_1,\dots,a_n))
\]
for all $(a_1,\dots,a_n)\in A^n$.
To simplify notation we often suppress the subscript
on fundamental operations and just write $\te f$ or $\te f(x_1,\dots,x_r)$.
Likewise, we often omit the algebra superscript on term operations.
We also use the bar convention by writing $\bc a$ for $(a_1,\dots,a_n)$.

The collection of all term operations on an algebra
forms a {\em clone}, that is, a family of operations on a set that
contains all the projection operations and is closed under composition.
Thus the set of term operations of $\m A$ is called the
{\em clone of $\m A$} and is denoted $\clo A$.
The clone of all operations on $A$ that can be obtained from the
term operations of $\m A$ and all the constant operations
is called the {\em clone of polynomial operations} of $\m A$ and
is denoted $\pol A$. The set of $n$-ary polynomial operations
of $\m A$ is written $\poln  n A$.
Two algebras $\m A$ and $\m B$ are said to be
{\em polynomially equivalent} if they have the same universe
and $\pol A = \pol B$.
A unary polynomial $\po e \in \poln 1 A$ is said to be idempotent
if $\po e(\po e(a)) = \po e(a)$ for all $a \in A$.

A {\em subuniverse} of an algebra $\m A$ is a set $S\subseteq A$
that is closed under the fundamental operations of $\m A$,
that is, $\te f(\bc a)\in S$ for
every fundamental operation $\po f$ of $\m A$ and every $\bc a \in S^r$.
An algebra $\m B$ is a {\em subalgebra} of $\m A$ if $\m B$ and $\m A$
have the same similarity type, the universe of $\m B$ is a subuniverse
of $\m A$, and for every operation symbol $\te f_i$, the operation
$\te f_i^{\m B}$ is the restriction to $B$ of the operation $\te f_i^{\m A}$.
Since the intersection of an arbitrary family of subuniverses of
an algebra $\m A$ is a subuniverse it follows that the set of all
subuniverses of $\m A$, denoted ${\rm Sub}\: \m A$,
forms a complete lattice when ordered by inclusion.
For a subset $X$ of $A$,
the {\em subuniverse of $\m A$ generated by $X$},
denoted ${\rm Sg}^{\m A}(X)$, is the intersection
of all subuniverses of $\m A$ that contain $X$.
Another way to describe ${\rm Sg}^{\m A}(X)$ is to observe
that the subuniverse generated by $X$ consists of all
elements of the form $\te t^{\m A}(\bc x)$ where
$\te t$ ranges over all terms for the language of $\m A$ and the
$\bc x$ are tuples from $X$.

Given two algebras $\m A$ and $\m B$ of the same similarity
type, a function $h:A \to B$ is called a {\em homomorphism}
if $h(\te f(a_1,\dots,a_r))=\te f(h(a_1),\dots,h(a_r))$ for every
fundamental operation $\te f$ of $\m A$ and all $a_i\in A$.
A homomorphism is an {\em isomorphism} if it is both
one-to-one and onto. If $h$ is a homomorphism from $\m A$
to $\m B$, then $h(A)$ is a subuniverse of $\m B$ and if
$A={\rm Sg}^{\m A}(X)$, then the subuniverse $h(A)$ is
generated by $h(X)$.

A congruence relation on an algebra $\m A$ is an equivalence
relation $\theta$ on $A$ that is preserved by the fundamental
operations of $\m A$, that is, if $f$ is an $r$-ary fundamental
operation and $(a_1,b_1),\dots,(a_r,b_r)\in \theta$,
then $(\te f(\bc a),\te f(\bc b)) \in \theta$.
Notation that is often used to express that $(a,b)$ is in the congruence
relation $\theta$ includes  $a \theta b$ and
$a \stackrel \theta \equiv b$.
For a congruence relation $\theta$ on $\m A$ the congruence class containing
an element $a$ is denoted $a/\theta$ and $A/\theta$ is the set
of all congruence classes of $\theta$.
The intersection of a family of congruence relations of an algebra
is again a congruence relation so the set of all congruence
relations of $\m A$, when ordered by inclusion, forms a complete
lattice.  The lattice of congruence relations of $\m A$ is denoted
$\Cn A$.  The top element of this lattice is $A\times A$ and is
written as $1_A$;  the bottom element is the diagonal
$0_A$, which consists of all pairs  $(a,a)$ for $a\in A$.
We frequently omit the subscripts in $0_A$ and $1_A$.

For a set $Z\subseteq A\times A$ the
{\em congruence relation on $\m A$ generated by Z}
is the intersection of all $\theta \in \cn A$ for which
$Z \subseteq \theta$. We write ${\rm Cg}^{\m A}(Z)$ for this
congruence relation but in the case that $Z=\set{(a,b)}$
we write ${\rm Cg}^{\m A}(a,b)$.
Like in the case of subuniverses ${\rm Sg}^{\m A}(X)$
there is an intrinsic way to describe the congruence ${\rm Cg}^{\m A}(Z)$.

\begin{lm}
\label{lm-malcev-chain}
Suppose that in an algebra $\m A$ we have
$(a,b)\in {\rm Cg}^{\m A}\left( Z \right)$ for some $Z\ci A^2$.
Then
\begin{itemize}
  \item there is a natural number $n$, a sequence $(y_1,z_1),\ldots,(y_n,z_n)$ of pairs in $Z$,
        a sequence of unary polynomials $\po p_1,\ldots, \po p_n$ of $\m A$
        and a sequence $x_0,\ldots,x_n$ of elements of $\m A$ such that
        \[
        \begin{array}{l}
        a=x_0, \ \  x_n=b \mbox{\ \ and \ }\\
        \set{x_{i-1},x_i} = \set{\po p_i(y_i),\po p_i(z_i)} \mbox{\ for all \ } 1\leq i \leq n,
        \end{array}
        \]
  \item if additionally $\m A$ is finite and has a ternary polynomial $\po d$
        that behaves like a Malcev operation on a subset $B \ci A$
        (i.e., $\po d(x,x,y) = y = \po d(y,x,x)$ for all $x,y\in B$)
        which is the range of a unary idempotent polynomial $\po e_B$ of $\m A$,
        then for $a,b \in B$ and $Z=\set{(c,d)}$
        there is a single unary polynomial $\po p$ with
        $\po p(c)=a$ and $\po p(d)=b$.
\end{itemize}
\end{lm}

\begin{proof}
The first item is due to Malcev. The second item is also a part of folklore, but we will include its proof for the reader convenience.

To see the second item note that from the first one
we know that $\m A$ has the unary polynomials
$\po p_1,\ldots, \po p_n$ and elements $x_0,\ldots,x_n$ such that
\[
a=x_0, b=x_n \mbox{\ and \ }
\set{x_{i-1},x_i} = \set{\po p_i(c),\po p_i(d)} \mbox{\ for all \ } 1\leq i \leq n
\]
and applying $\po e_B$ to this chain we may assume that the ranges of the $\po p_i$'s are contained in $B$, so that the entire chain of the $x_i$'s lives in $B$.
First look at $\set{x_{i-1},x_i} = \set{\po p_i(c),\po p_i(d)}$.
If $(x_{i-1},x_i) = (\po p_i(d), \po p_i(c))$, replace $\po p_i(x)$ by
$\po d_B(\po p_i(c), \po p_i(x), \po p_i(d))$,
so that after this replacement we have
$(x_{i-1},x_i) = (\po p_i(c), \po p_i(d))$ for all $i$.
Now we will show that if $n>1$ such sequence can be shortened and this additional requirements are kept.
Indeed, for $\po p_{1,2}(x) = \po d_B(\po p_1(x), \po p_1(d), \po p_2(x))$
we have $(\po p_{1,2}(c),\po p_{1,2}(d)) = (x_0,x_2)$.
\end{proof}

\medskip
Some terminology from lattice theory is used in describing $\Cn A$.
For  $a \leq b$ in a lattice $\m L$
the ordered pair $(a,b)$ is called a {\em quotient} in $\m L$ and
the {\em interval} from $a$ to $b$,
written $\intv a b$, is the subuniverse of $\m L$ consisting of
$\set{c \in L:a \leq c \leq b}$.
The element $a$ is {\em covered} by $b$ if $a<b$ and $I[a,b]=\set{a,b}$.
If $a$ is covered by $b$, then
we write $a \prec b$   and call $I[a,b]$
a {\em prime interval} or a {\em prime quotient}.
A {\em subcover} of an element $b$ is any
element covered by $b$.
An {\em atom} in a lattice with least element 0 is any element
that covers 0 and a {\em coatom}  or {\em dual atom}
in a lattice with largest element 1 is any element covered by 1.
If $I[a,b]$ and $I[c,d]$ are intervals such that
$b \land c = a$ and $b \lor c=d$, then
$I[a,b]$ is said to {\em transpose up} to
$I[c,d]$, written $I[a,b] \nearrow I[c,d]$;
and $I[c,d]$ is said to {\em transpose down} to $I[a,b]$,
written $I[c,d] \searrow I[a,b]$; and the two intervals
are called {\em transposes} of one another.
Two intervals are said to be {\em projective}
if one can be obtained from the other by a finite sequence
of transposes.
A fundamental fact in lattice theory is that a lattice
is modular if and only if its projective intervals are isomorphic.
Another equivalent condition for modularity is that the lattice
has no elements $a,b,c$ satisfying $a < b, a\lor c = b \lor c$
and $a \land c = b \land c$.  Such a 5-element sublattice generated
by $a,b,c$  will be called an $[a,b,c]$-{\em pentagon}.

An algebra $\m A$
is {\em simple} if it is nontrivial and $\cn A$ consists solely
of $1_A$ and $0_A$.
An algebra is called {\em congruence distributive}
or {\em congruence modular} if its congruence lattice satisfies
the distributive identity or the modular identity.
Two congruence relations $\theta,\tau \in \cn A$
{\em permute} if $\theta \circ \tau= \tau \circ \theta$.  If
$\theta$ and $\tau$ permute, then
$\theta \lor \tau = \theta \circ \tau$
in $\Cn A$.
An algebra is {\em congruence permutable} if every pair of its congruence relations
permute.

\bigskip
Homomorphisms and congruence relations are naturally linked:
If $h$ is a homomorphism on $\m A$, then the {\em kernel of $h$},
denoted ${\rm ker}(h)$, is the set of all
$(a_1,a_2)\in A^2$ for which $h(a_1)=h(a_2)$.
For every homomorphism $h$ the relation ${\rm ker}(h)$
is a congruence on $\m A$.
On the other hand, if $\theta \in \cn A$, then the congruence classes
of $\theta$ form the elements of an algebra $\m A/\theta$ and
the map $a\mapsto a/\theta$ is a homomorphism from $\m A$ onto $\m A/\theta$
with kernel $\theta$.

\bigskip
We next consider direct products of algebras.
Suppose $\m A_j$, for $j \in J$, are algebras of the same similarity
type indexed by a set $J$.
The {\em direct product} of these algebras,
denoted  $\prod_{j\in J}\m A_j$,
is an algebra of the same similarity type as the $\m A_j$
with universe $\prod_{j\in J} A_j$ and fundamental operations defined
coordinatewise:
$\te f(\bc a, \bc b, \bc c, \dots)_j=\te f(a_j,b_j,c_j,\dots)$
for all $j\in J$.
Often the index set $J$ is finite, say $J=\set{1,\dots,n}$,
and we write $\m A_1\times \dots\times\m A_n$ for the direct product
in this situation.
If $J$ is the empty set, then $\prod_{j\in J}\m A_j$ is
a trivial algebra.
A direct product of copies of a single algebra $\m A$ is called
a {\em direct power} of $\m A$.
We write $\m A^J$ for a direct power of $\m A$ indexed by a set $J$
and we often view the elements of this algebra as functions from $J$ to $A$.

If $\m A = \prod_{j\in J}\m A_j$,
then the {\em $j$-th projection}
map $\pi_j$ is a homomorphism of $\m A$ onto $\m A_j$.
The kernel of $\pi_j$ is usually written as $\eta_j$
and thus for $a,b \in A$
we have $(a,b)\in \eta_j$ if and only if $a(j)=b(j)$.
The $\eta_j$ are called {\em projection kernels}.
It is easily
checked that if $J_1$ and $J_2$ are nonvoid complementary
subsets of  $J$ and
$\alpha_i=\bigwedge_{j\in J_i} \eta_j$ for $i=1,2$,
then in $\Cn A$ we have:
\begin{tenumerate}
\item the congruences $\alpha_1$ and $\alpha_2$ permute,
\item $\alpha_1 \lor \alpha_2 =1_A$,
\item $\alpha_1 \land \alpha_2 = 0_A$.
\end{tenumerate}
Conversely, if $\m A$ is any algebra and $\alpha_1$ and $\alpha_2$
are any two congruences for which these three conditions hold,
then $\alpha_1$ and  $\alpha_2$ are each called
{\em factor congruences} of $\m A$
and $\m A \simeq \m A/\alpha_1\times \m A/\alpha_2$.
An algebra is called {\em directly indecomposable} if it is nontrivial
and is not isomorphic to the direct product of two nontrivial
algebras.  Every finite algebra is isomorphic to the direct
product of directly indecomposable algebras but this is not necessarily
the case for infinite algebras.

Certain subalgebras of a direct product called subdirect
products play an important role in our work.
An algebra $\m A$ is a {\em subdirect product} of the algebras
$\m A_j$, for $j\in J$, if $\m A$ is a subalgebra of $\prod_{j\in J} \m A_j$
and for each $j\in J$ the projection map from $\m A$ to $\m A_j$
is onto.  Thus, if $\m A$ is a subdirect product of the $\m A_j$
for $j\in J$ and $\gamma_j$ is the kernel of
the $j$-th projection homomorphism from $\m A$ to $\m A_j$,
then $\bigcap_{j\in J} \gamma_j=0_A$ and each
$\m A_j$ is isomorphic to $\m A/\gamma_j$.
Conversely, if a family of congruence
relations, $\gamma_j$ for $j\in J$,
on an algebra $\m A$ has the property that
$\bigcap_{j\in J} \gamma_j=0_A$, then $\m A$ is isomorphic to
an algebra that is the subdirect product of $\m A/\gamma_j$ for $j\in J$.
A {\em subdirect representation} of $\m A$ with subdirect
factors $\m A_j$ is a homomorphic embedding $h$
of $\m A$ into $\prod_{j\in J}\m A_j$ for which $h( A)$
is the universe of an algebra that
is a subdirect product of the $\m A_j$.

\medskip
An algebra $\m A$ is {\em subdirectly irreducible} if it is nontrivial
and in any subdirect representation of $\m A$  at least
one of the projection maps is an isomorphism.
We use the following internal characterization of a subdirectly
irreducible algebra:  An algebra $\m A$ is subdirectly irreducible
if and only if there is a $\mu \in \cn A$ such that $ 0_A < \mu $
and $\mu \leq \theta$ for all $0_A < \theta \in \cn A$.
The congruence relation $\mu$ is called the {\em monolith} of the
subdirectly irreducible algebra $\m A$.
Thus, $\m A$ is subdirectly irreducible if and only if in the
lattice $\Cn A$ the element $0_A$ is strictly meet irreducible.
A theorem of Birkhoff states that every algebra is a subdirect product
of subdirectly irreducible algebras. This theorem is equivalent to the
statement that in any algebra $\m A$ the congruence relation $0_A$
is the intersection of all strictly meet irreducible members of
$\Cn A$.  If $\m A$ is  a subdirectly irreducible algebra with monolith
$\mu$, then $\mu=\cg A a b$ for every $(a,b)\in \mu - 0$.

\medskip
The class of algebras is a {\em variety} if it is closed under taking homomorphic images, subalgebras and products of algebras.
A classic preservation theorem of Birkhoff states that
a class of algebras is a variety if and only if it is an
equational class, i.e. a class consisting of all algebras satisfying a certain set of identities.

A variety is {\em congruence distributive, congruence modular,} or
{\em congruence permutable} if every algebra in the variety is
congruence distributive, congruence modular, or has permuting
congruence relations, respectively.
Obviously congruence distributive varieties are congruence modular.
But also congruence permutable varieties are known to be congruence modular.
A classic result of Malcev states that a variety $\vr V$
is congruence permutable if and only if $\vr V$ has a ternary term
$\te d$ for which $\vr V \models \te d(x,x,y) \approx \te d(y,x,x) \approx y$.
Such a term is called a {\em Malcev term} for $\vr V$.

There are also several characterizations (due to B.~J\'onsson, A.~Day or H.P.~Gumm)
of congruence modular or distributive varieties in terms of identities they have to satisfy.
We will use one such characterization via so called {\em directed Gumm terms}
which is described in \cite{kazda-directed}.

\begin{thm}
\label{thm-gumm}
A variety $\vr V$ is congruence modular if and only $\vr V$ has ternary terms
\[
\d_1(x,y,z),\ldots, \d_n(x,y,z), \q(x,y,z)
\]
satisfying the following equalities:
\[
\begin{array}{rclclcll}
x &=&  \d_i(x,y,x), &&                  &&&\mbox{\em for all $i=1,\ldots,n$,}  \\
x &=&  \d_1(x,x,y), &&                  &&&\\
    && \d_i(x,y,y) &=& \d_{i+1}(x,x,y), &&&\mbox{\em for all $i=1,\ldots,n-1$,}\\
    && \d_n(x,y,y) &=& \q(x,y,y),       &&&\\
    &&              && \q(x,x,y)        &=& y.&
\end{array}
\]
\end{thm}

\bigskip
The material on universal algebra presented so far
is ``classical" and was all known by the mid 1960s.
In our work we will require some deep results that have
come out of two more recent developments:
generalized commutator theory and tame congruence theory.
We now present the basics of these two topics.

Fuller discussions of the generalized commutator may be found in
\cite{fm}, \cite[Section 4.13]{mmt} and \cite[Chapter 3]{hm}.
The main reference for tame congruence theory is \cite{hm}.

We begin with the theory of the commutator.
Let $\m A$ be an algebra, $\gamma \in \cn A$, and $R,S\subseteq A^2$.
We say {\em $R$ centralizes $S$ modulo $\gamma$}, denoted
$C(R,S;\gamma)$, if for every $n \geq 1$, every $(n+1)$-ary
term $\te t$, every $(a,b) \in R$, and every
$(c_1,d_1),\dots,(c_n,d_n)\in S$ we have
\[
\te t(a,\bc c) \congruent{\gamma} \te t(a,\bc d)  \mbox{\ \  iff \ \  }
\te t(b,\bc c) \congruent{\gamma} \te t(b,\bc d).
\]

The following facts are easily verified.

\begin{prp} \label{prp-basic-central1}
For binary relations that are congruence relations on $\m A$:
\begin{tenumerate}
\item
If $\alpha' \subseteq \alpha$ and $\beta' \subseteq \beta$,
then $C(\alpha,\beta;\gamma)$ implies $C(\alpha',\beta';\gamma)$.
\item
If $C(\alpha,\beta,\gamma_i)$ for all $i\in I$,
then $C(\alpha,\beta;\bigcap_{i\in I}\gamma_i)$.
\item
$C(\alpha,\beta;\alpha)$ and $C(\alpha,\beta;\beta)$.
\item
If $C(\alpha_i,\beta,\gamma)$ for all $i\in I$,
then $C(\bigvee_{i\in I}\alpha_i,\beta;\gamma)$.
\item
If $\theta \ci \alpha,\beta, \gamma$ then $C(\alpha,\beta;\gamma)$
holds in $\m A$ iff $C(\alpha/\theta,\beta/\theta;\gamma/\theta)$ holds in the quotient $\m A/\theta$.
\end{tenumerate}
Moreover, {\rm (1)} and {\rm (2)} hold for arbitrary binary relations
$\alpha, \alpha',\beta,\beta'$, and
{\rm (3)} holds if $\alpha$ and $\beta$  are binary relations
that are preserved by the fundamental operations of~$\m A$.
\end{prp}

An algebra $\m A$ is {\em Abelian}, or is said to satisfy the
{\em term condition}, if $C(1_A,1_A;0_A)$ holds.
Note that if $C(1_A,1_A;\gamma)$, then $\m A/\gamma$ is Abelian.

If $\alpha$ and $\beta$ are congruence relations on an algebra
$\m A$, then the {\em commutator} of $\alpha$ and $\beta$, denoted
$\com \alpha \beta$, is the least congruence $\gamma$ for which
$C(\alpha,\beta;\gamma)$.  The {\em centralizer} of
$\beta$ {\em modulo} $\alpha$, denoted $\centr \beta \alpha$, is the
largest congruence $\delta$ for which $C(\delta,\beta;\alpha)$.

We will appeal, often without reference, to the following facts about the
centralizer and the commutator:
\begin{prp}
\label{prp-basic-central2}
For congruence relations in an arbitrary algebra $\m A$
\begin{tenumerate}
\item
$C(\alpha, \beta;\gamma)$ if and only if $\alpha \leq \centr \beta \gamma$,
\item
$\centr {0_A} \alpha = 1_A$,
\item
$\alpha \leq \centr \beta \alpha$.
\setcounter{entmp}{\value{enumi}}
\end{tenumerate}
If $\m A$ belongs to a congruence modular variety then we additionally have {\rm (see \cite{fm})}
\begin{tenumerate}
\setcounter{enumi}{\value{entmp}}
\item
$\com \alpha \beta = \com \beta \alpha$,
\item
$C(\alpha,\beta;\gamma)$ if and only if $\com \alpha \beta \leq \gamma$,
\item
$\com \alpha {\bigvee_i \beta_i} =\bigvee_i\com \alpha {\beta_i}$,
\item
$\centr {\bigvee_i \beta_i} \alpha = \bigwedge_{\; i} \centr {\beta_i} {\alpha}$,
\item
$\centr \beta {\bigwedge_{\; i} \alpha_i} = \bigwedge_{\; i} \centr \beta {\alpha_i}$,
\item
if the intervals $I[\alpha_1,\beta_1]$ and $I[\alpha_2,\beta_2]$ are projective
in the  lattice $\Cn A$, then
$\centr {\beta_1}{\alpha_1} =\centr{\beta_2}{\alpha_2}$.
\end{tenumerate}
\end{prp}

\medskip
\noindent
A consequence of items (2) and (3) in
Proposition \ref{prp-basic-central1} is that
$\com \alpha \beta \leq \alpha \cap \beta$
for all congruence relations in
an arbitrary algebra,
however, for algebras in congruence distributive varieties
it is known that $\com \alpha \beta = \alpha \cap \beta$,
see e.g., \cite[p. 258] {mmt}.

By means of the commutator it is possible to define notions of
Abelian, solvable and nilpotence for arbitrary algebras.
Let $\alpha \leq \beta$ be congruence relations of an
algebra $\m A$.
The congruence relation $\beta$ is {\em Abelian over} $\alpha$
if $C(\beta,\beta;\alpha)$
and $\beta$ is {\em Abelian}
if $C(\beta,\beta;0_A)$.
We say $\beta$ is {\em solvable over} $\alpha$
if there exists a finite chain
of congruence relations $\beta=\gamma_0 \geq \gamma_1 \geq \dots \geq \gamma_m=\alpha$
such that $\gamma_i$ is Abelian over $\gamma_{i+1}$ for all $i<m$.
A congruence relation $\beta$
is {\em solvable} if it is solvable over $0_A$.
An algebra $\m A$ is {\em solvable} if $1_A$, and hence every congruence
relation of $\m A$, is solvable.  \
An algebra $\m A$ is {\em locally solvable}
if every finitely generated subalgebra of $\m A$ is solvable.
It can be argued that in the congruence lattice of a finite algebra $\m A$
the join of all the solvable congruence relations is itself
solvable.  This largest solvable congruence relation is called
the {\em solvable radical} of $\m A$.

For a congruence $\theta$ and $i=1,2,\dots$ we write
\[ \begin{array}{rclcrcl}
\theta^{(1)}&=&\theta & \ \ \ & \theta^{[1]}&=&\theta \\
\theta^{(i+1)}&=&[\theta,\theta^{(i)}] & \ \ \ &\theta^{[i+1]}&=&[\theta^{[i]},\theta^{[i]}].
\end{array}
\]

A congruence relation $\theta$ on $\m A$ is called
{\em $k$-step left nilpotent}
if $\theta^{(k+1)}=0_A$ and the algebra $\m A$
is {\em left nilpotent} if $1_A$ is $k$-step left nilpotent for some finite $k$.
In the congruence modular varieties we use the word nilpotent rather than left nilpotent.
Note that $\theta$ is solvable if $\theta^{[k]}=0_A$ for some $k$.

The following strengthening of the nilpotency is also relevant in our setting.
First, for a bunch of congruences
$\alpha_1,\ldots,\alpha_k,\beta,\gamma \in \con A$
we say that $\alpha_1,\ldots,\alpha_k$ centralize $\beta$ modulo $\gamma$,
and write $C(\alpha_1,\ldots,\alpha_k,\beta;\gamma)$,
if for all polynomials $\po f \in \pol A$ and all tuples
$\o a_1 \congruent{\alpha_1} \o b_1, \ldots, \o a_k \congruent{\alpha_k} \o b_k$
and $\o u \congruent{\beta} \o v$
such that
\[
\po f(\o x_1,\ldots, \o x_k, \o u) \congruent{\gamma} \po f(\o x_1,\ldots, \o x_k, \o v)
\]
for all possible choices of
$(\o x_1,\ldots, \o x_k)$ in $\set{\o a_1,\o b_1} \times \ldots \times \set{\o a_k,\o b_k}$
but $(\o b_1,\ldots.\o b_k)$,
we also have
\[
\po f(\o b_1,\ldots, \o b_k, \o u) \congruent{\gamma} \po f(\o b_1,\ldots, \o b_k, \o v).
\]
This notion was introduced by A.~Bulatov \cite{bulatov:supercomm}
and further developed by E.~Aichinger and N.~Mudrinski
\cite{aichinger-mudrinski}.
In particular they have shown that for all $\alpha_1,\ldots,\alpha_k \in \con A$
there is the smallest congruence $\gamma$ with $C(\alpha_1,\ldots,\alpha_k;\gamma)$
called the $k$-ary commutator and denoted by $\commm{\alpha_1}{\alpha_k}$.
Such generalized commutator behaves especially well in algebras from congruence modular varieties.
In particular this commutator is monotone, join-distributive and we have
\[
\comm{\alpha_1} {\commm{\alpha_2}{\alpha_k}} \leq \commm{\alpha_1}{\alpha_k}
\]
Thus every $k$-supernilpotent algebra, i.e.~algebra satisfying
$[ \overbrace{1,\ldots,1}^{\text{\scriptsize $k\!+\!1$ times}} ] =0$,
is $k$-nilpotent.
The following properties, that can be easily inferred from
the deep work of R.~Freese and R.~McKenzie \cite{fm}
and K.~Kearnes \cite{kearnes:small-free}, have been summarized in \cite{aichinger-mudrinski}.

\begin{thm}
\label{thm-prime-supernil}
For a finite algebra $\m A$ 
from a congruence modular variety the following conditions are equivalent:
\begin{tenumerate}
\item
$\m A$ is $k$-supernilpotent,
\item
$\m A$ is $k$-nilpotent, decomposes into a direct product of algebras of prime power order
and the clone $\clo A$ is generated by finitely many operations,
\item
$\m A$ is $k$-nilpotent and all commutator polynomials have rank at most $k$.
\end{tenumerate}
\end{thm}

The commutator polynomials mentioned in condition (3) of Theorem \ref{thm-prime-supernil}
are the paradigms for the failure of supernilpotency.
We say that $\po t(x_1, \ldots, x_{k-1}, z) \in \poln k A$ is a commutator polynomial of rank $k$ if
\begin{itemize}
  \item $\po t(a_1,\ldots,a_{k-1},b)=b$ whenever $b \in \set{a_1,\ldots,a_{k-1}} \ci A$,
  \item $\po t(a_1,\ldots,a_{k-1},b)\neq b$ for some $a_1,\ldots,a_{k-1},b \in A$.
\end{itemize}

\medskip
We next sketch the material on tame congruence theory that we will need.

For a nonvoid subset $U$ of an algebra $\m A$ the
{\em algebra induced by $\m A$ on U} is the algebra
$\rst {\m A} U$ whose universe is $U$ and whose fundamental
operations are all polynomials $\po p \in \poln m A$ for which
$\rst {\po p} {U^m}$ maps $U^m$ into $U$.
The algebra $\rst {\m A} U$ is nonindexed, that is, there is
no index set specified for the set of fundamental operations.
Note that every polynomial  operation of $\rst {\m A} U$ is its fundamental operation.
Two nonvoid subsets $U$ and $V$ of $\m A$ are
called {\em polynomially isomorphic} if there exist
$\po f,\po g \in \poln 1 A$ such that $\po f(U)=V$, $\po g(V)=U$,
$\po f\po g$ is the identity on $V$, and $\po g\po f$ is the identity on $U$.
If $U$ and $V$ are polynomially isomorphic, then the algebras
$\rst {\m A} U$ and $\rst {\m A} V$ are isomorphic
as nonindexed algebras,
that is, it is possible to index the fundamental operations
of each with one index set
so that the resulting algebras are isomorphic in the usual sense.

An {idempotent} polynomial for an algebra $\m A$ is any
$\po e \in \poln 1 A$ such that $\po e^2(x)=\po e(x)$ for all $x\in A$.
For an idempotent polynomial $\po e$ the restriction
$\rst {\po e} {\po e(A)}$ is the identity map on $\po e(A)$.
Algebras induced by $\m A$ on the range of an idempotent
polynomial have a particularly simple characterization for their
fundamental operations.  Namely, if $\po e$ is idempotent for
$\m A$ and $U=\po e(A)$,
then the fundamental operations of $\rst {\m A} U$ consist
of all polynomials of the form $\rst {\po e \po p} {U}$
where $\po p$ ranges over all polynomials of $\m A$.
The collection of all idempotent polynomials for $\m A$ is
denoted $E(\m A)$.

Let $\alpha < \beta$ in the congruence lattice of a finite
algebra $\m A$.  By ${\rm U}_{\m A}(\alpha,\beta)$ we denote all sets
of the form $\po f(A)$,
with at least two elements,
where $\po f\in \poln 1 A$ and
$\po f(\beta)\not \subseteq \alpha$.
Minimal members of ${\rm U}_{\m A}(\alpha,\beta)$,
that is, minimal when ordered by inclusion,
are called $(\alpha, \beta)$-{\em minimal sets of}
$\m A$. The set of all $(\alpha, \beta)$-minimal sets of
$\m A$ is denoted ${\rm M}_{\m A}(\alpha,\beta)$.

In a finite algebra $\m A$
a quotient $(\alpha,\beta)$ in $\Cn A$
is called {\em tame} if there exist
$V \in {\rm M}_{\m A}(\alpha,\beta)$
and $\po e \in E(\m A)$ such that $\po e(A)=V$ and for all
$\gamma \in \cn A$ if $\alpha < \gamma < \beta$, then
$\rst \gamma V \ne \rst \alpha V$ and $\rst \gamma V \ne \rst \beta V$.
Note that every prime quotient is tame.
A basic result in tame congruence theory is that
if $(\alpha,\beta)$ is a tame quotient, then all
$(\alpha, \beta)$-minimal sets of
$\m A$ are polynomially isomorphic.
If $(\alpha,\beta)$ is tame and $U\in {\rm M}_{\m A}(\alpha,\beta)$,
then any set of the form $ a/\beta\cap U$ that is not
of the form $ a/\alpha \cap U$ is called a {\em trace} of $U$ and an
$(\alpha,\beta)$-{\em trace} of $\m A$.
The union of all $(\alpha,\beta)$-traces of $U$ is called
the {\em body} of $U$ and those elements of $U$ not in the body
of $U$ form the {\em tail} of $U$.
If $N$ is a trace for $U$, then $\rst \alpha N$ denotes
$\alpha \cap N^2$, and $\rst \alpha N$ is a congruence on
the nonindexed algebra $\rst {\m A} N$.

\medskip
The interest in tame congruence theory in tame quotients and
their minimal sets and traces arises from the fact that
the local behavior of a tame quotient falls into one of
five distinct situations.  More specifically,
for any finite algebra $\m A$, for any tame quotient
$(\alpha, \beta)$, and for any trace $N$ of
$U\in {\rm M}_{\m A}(\alpha,\beta)$, the quotient algebra
$(\rst {\m A} N)/(\rst \alpha N)$ must be
polynomially equivalent to one of the following five types of algebras:

\begin{enumerate}
\item[{\tn 1}.] a {\rm G}-set,
\item[{\tn 2}.] a finite dimensional vector space over a finite
field,
\item[{\tn 3}.] a 2-element Boolean algebra,
\item[{\tn 4}.] a 2-element distributive lattice,
\item[{\tn 5}.] a 2-element semilattice.
\end{enumerate}

\noindent
Moreover, the particular type \tn 1, \tn 2, \tn 3, \tn 4, or
\tn 5 is independent of the choice of
$U$ and $N$.  This is called the {\em type} of the tame quotient
$(\alpha,\beta)$ and is denoted $\typ(\alpha,\beta)$.

The type of a tame quotient in a finite algebra
has significant consequences for local behavior and for
the algebraic structure of the algebra and the quotient.
For example, it is known that for a tame quotient $(\alpha,\beta)$,
$\typ(\alpha,\beta)\in\set{\tn 1, \tn 2}$ if and only if $\beta$ is Abelian over $\alpha$.
Because of this, types \tn 1 and \tn 2 are referred to as the {\em Abelian types}
and types \tn 3, \tn 4, and \tn 5 are the {\em non-Abelian types}.

In our work the tame quotients that we consider are usually prime quotients.
The following terminology is used in connection with the set of types of prime quotients in finite algebra.

For $\alpha \prec \beta$ the fact $\typ(\alpha,\beta)= \tn i$ will be sometimes denoted by
$\alpha \prect i \beta$.
For $\gamma < \delta$ in $\Cn A$ the set of all types $\typ(\alpha,\beta)$
for $\gamma \leq \alpha \prec \beta \leq \delta$ is denoted $\typ\set{\gamma,\delta}$.
The {\em type set of a finite algebra} $\m A$, denoted $\typ\set{\m A}$, is $\typ\set{0_A,1_A}$.
The {\em type set of a class} $\vr K$ of algebras consists of the union
of the type sets of the finite algebras in $\vr K$ and is
denoted $\typ\set{\vr K}$.

Two preservation theorems involving the calculus of types
that we will frequently use are that type is preserved under homomorphism
and that projective prime quotients have the same type, that is, for a finite algebra  $\m A$,
\begin{itemize}
\item
if $\delta\leq\alpha\prec\beta$ in $\Cn A$, then $\typ(\alpha,\beta)=
\typ(\alpha/\delta,\beta/\delta)$ in $\m A/\delta$,
\item
if $\alpha_1\prec\beta_1$ and $\alpha_2\prec\beta_2$ are projective prime quotients in $\Cn A$,
then $\typ(\alpha_1,\beta_1)=\typ(\alpha_2,\beta_2)$.
\end{itemize}
These two results show that if $\tn i \in \typ\set{\m A}$,
then there is a subdirectly irreducible algebra ${\m A}'$ with monolith
$\mu$ such that $\m A^\prime$ is a homomorphic image of $\m A$ and
$\typ(0_{A^\prime},\mu)=\tn i$.
Many of our arguments involve an analysis of $(0,\mu)$-minimal
sets and traces for such a monolith $\mu$.

We next summarize some of the algebraic properties that are
consequences of a prime quotient having a particular type.
Consider an arbitrary finite algebra $\m A$ with $\alpha \prec \beta$.
Let $U \in M_{\m A}(\alpha,\beta)$ and let $N$ be an
$(\alpha ,\beta)$-trace contained in $U$.
Suppose $\typ(\alpha,\beta)=\tn 3 \mbox{ or } \tn 4$.
For these two types,  it is known that $N$ is the unique $(\alpha,\beta)$-trace contained in $U$,
$\alpha|_N = 0$ and the algebra $\rst {\m A} {N}$ is polynomially equivalent to
a 2-element Boolean algebra or 2-element distributive lattice.
In both cases there are two binary polynomial $\land, \lor \in \poln 2 A$
such that $\rst \land U$ is a {\em pseudo-meet} and $\rst \lor U$ is a {\em pseudo-join},
\cite [Definition 4.18]{hm}.
This in particular means that we can label the two elements of $N$ with
$0$ and $1$ so that $\langle \set{0,1}, \rst \land N, \rst\lor N \rangle$ is a distributive
lattice with $0 < 1$.
Thus, every $n$-ary operation on $N=\set{0,1}$
that preserves this order is of the form $\rst {\po p} N$
for some $\po p \in \poln n A$.
If $\typ(\alpha, \beta)=\tn 3$, then in addition to the binary
polynomials $\land$ and $\lor$ that we have in the type \tn 4 case,
there is also a unary polynomial $'$ such that $0'=1, 1'=0$ and $A' = U$.
The algebra $\langle \set{0,1}, \rst \land N, \rst\lor N, \rst ' N \rangle$
is a Boolean algebra and thus every $n$-ary operation on $N$
is the restriction to $N$ of some $n$-ary polynomial on $\m A$
that can be built using $\land, \lor,$ and $'$.

If $\typ(\alpha,\beta)=\tn 2$, then there may be more than one
trace contained in $U$.  Let $B$ be the body of $U$.
A useful result (see \cite [Definition 4.22]{hm}) that applies to this type \tn 2 case
is that there is a $\po d \in \poln 3 A$ such that
\begin{tenumerate}
\item
$\po d(x,x,x)=x$ for all $x \in U$.
\item
$\po d(x,x,y)=y=\po d(y,x,x)$ for all $x \in B$ and $y\in U$.
\item
For every $a,b \in B$, the unary polynomials given by
$\po d(x,a,b), \po d(a,x,b)$, and $\po d(a,b,x)$ are
permutations of $U$
\item
$B$ is closed under $\po d$, that is, $\po d(a,b,c)\in B$
for all $a,b,c\in B$.
\end{tenumerate}
The polynomial $\po d$ is called a {\em pseudo-Malcev operation} for $U$.

\bigskip

Since we are particularly interested in finite algebras from congruence modular varieties
we conclude our discussion of tame congruence theory
by citing some results that connect it with the theory of the generalized commutator
in {\em locally finite varieties},
i.e. varieties in which finitely generated algebras are finite.

\begin{thm}
\label{thm-cm-types}
Let $\vr V$ be a locally finite variety.
\begin{tenumerate}
\item
$\vr V$ is congruence modular if and only if
$\typ\set{\vr V} \ci \set{\tn 2, \tn 3, \tn 4}$
and minimal sets in finite algebras of $\vr V$ have empty tails.
\item
$\typ\set{\vr V}\subseteq \set{\tn 2}$ if and only if $\vr V$
is congruence permutable and every algebra in $\vr V$ is locally solvable.
\end{tenumerate}
\end{thm}

The varietal conditions given in item (2) of Theorem \ref{thm-cm-types}
will be of special interest in our work.
A variety $\vr V$ is called {\em affine} if it is congruence modular and Abelian.
It can be argued that if $\vr V$ is affine then it is also congruence permutable.
The  properties of affine varieties are developed in \cite{fm}.
Each affine variety $\vr V$ has a corresponding ring $\m R$ with unit
such that every algebra in $\vr V$ is polynomially equivalent to an $\m R$-module
and conversely every $\m R$-module is polynomially equivalent to an algebra in $\vr V$.

\ssection{Some easy observation
\label{sect-ex}}

Except canonical \npc problems (like \textsc{SAT} or $k$-colorability of graphs) used in our proofs of \npc{}ness we will also need the following easy observation,
a straightforward proof of which can be found in \cite{gorazd-krzacz:preprimal}.

\begin{prp}
\label{prp-dl01}
It is \npc to decide whether the following systems of two equations of the form
\begin{eqnarray*}
  \mmeet_{i=1}^m x^i_1 \join x^i_2 \join x^i_3 &=& 1,\\
  \jjoin_{i=1}^n y^i_1 \join y^i_2 \join y^i_3 &=& 0,
\end{eqnarray*}
where $x^i_j$ and $y^i_j$ are variables,
have  solutions in the $2$-element lattice.
\myqed
\end{prp}

We continue this section with the proofs of Proposition \ref{prp-csp-csat}
and Example \ref{ex-non-hom}.

\medskip

{\textsc{Proposition \ref{prp-csp-csat}.}}
{\em
For every finite relational structure $\rel D$ (with finitely many relations)
there is a finite algebra $\mrel A D$
such that the problem $\csp D$ is polynomially equivalent to $\csat {A[\rel D]}$.
}

\medskip
\begin{proof}
Without loss of generality we may assume that $\rel D$ has both satisfiable and unsatisfiable instances, say $\top$ and $\bot$, respectively.
Now, for the relational structure $\rel D = \left(D, \mathcal{R}\right)$
put $\mrel A D$ to be $\left(A; \meet, \set{f_R}_{R\in{\mathcal R}}  \right)$,
where\begin{itemize}
       \item $A= D\cup \set{0,1}$ with $0,1 \not\in D$,
       \item the binary operation $\meet$ is defined by:
       \[
       a \meet b =
       \left\{
       \begin{array}{ll}
       1, & \mbox{\rm if $a=1=b$,}\\
       0, & \mbox{\rm otherwise,}
       \end{array}
       \right.
       \]
       \item $f_R$ is the $\set{0,1}$-characteristic function of the relation $R$, i.e.
       \[
       f_R(a_1,\ldots,a_k) =
       \left\{
       \begin{array}{ll}
       1, & \mbox{\rm if $(a_1,\ldots,a_k)\in R$,}\\
       0, & \mbox{\rm otherwise.}
       \end{array}
       \right.
       \]
     \end{itemize}
It should be obvious that the instance
\begin{equation}
\label{eq-meet-r}
R_1(x^1_1,\ldots,x^1_{k_1}) \meet \ldots \meet R_s(x^s_1,\ldots,x^s_{k_s})
\end{equation}
of $\csp D$ transforms equivalently to the following instance of $\csat {A[\rel D]}$
\[
f_{R_1}(x^1_1,\ldots,x^1_{k_1}) \meet \ldots \meet f_{R_s}(x^s_1,\ldots,x^s_{k_s}) = 1.
\]

On the other hand the only polynomials of $\mrel A D$ that are non constant
are among those that have the following form:
\begin{equation}
\label{eq-meet-f}
x^0_1 \meet \ldots x^0_{k_0} \meet
f_{R_1}(x^1_1,\ldots,x^1_{k_1}) \meet \ldots \meet f_{R_s}(x^s_1,\ldots,x^s_{k_s})
\meet 1,
\end{equation}
where the last conjunct (namely $1$) may be absent
and the $x^0_j$'s are not among the $x^i_j$'s with $i \geq 1$.
Moreover the range of such polynomials is contained $\set{0,1}$
where the value $1$ is obtained by sending all the $x^0_j$'s to $1$
and the other variable to the values in $D$ satisfying (\ref{eq-meet-r}).
Thus the only nontrivial instances (i.e. the ones that do not transform to $\top$ or $\bot$)
of $\csat {A[\rel D]}$ have the form $\po t(\o x) = 1$,
with $\po t(\o x)$ being described in (\ref{eq-meet-f}).
Such an equation obviously translates to the equivalent instance (\ref{eq-meet-r}) of $\csp D$.
\end{proof}

\bigskip

\textsc{Example \ref{ex-non-hom}.}
{\em
There is a finite algebra $\m A$ and its congruence $\theta$ such that
$\csat A \in \ptime$ while $\csat {A/\theta}$ is \npc.
}
\medskip
\begin{proof}
The operations of the algebra $\m A$ will be defined in such a way
that the satisfiability of the polynomial equation $\po t(\o x) = \po s(\o x)$
easily reduces to the one over the $2$-element lattice
whenever the ranges of the polynomials $\po t$ and $\po s$ are not disjoint.
On the other hand it will be possible to define the congruence $\theta$
and lattice-like polynomials $\po l$ and $\po r$
so that the one equation of the form $\po l(\o x) = \po r(\o x)$
encodes modulo $\theta$ a system of two lattice equations.

For the underlying set of our algebra we put
\[
A=\set{0,1,0',1',
0_{\meet,l},1_{\meet,l},s_{\meet,l},
0_{\meet,r},1_{\meet,r},s_{\meet,r},
0_{\join,l},1_{\join,l},s_{\join,l},
0_{\join,r},1_{\join,r},s_{\join,r},
}.
\]

Our basic operations come in two sorts:
\begin{itemize}
  \item left: ternary `disjunction' $D_l$, binary `conjunction' $\meet_l$ and unary $f_l$
  \item right: ternary `conjunction' $C_r$, binary `disjunction' $\join_r$ and unary $f_r$,
\end{itemize}
To define these operations we will refer to an external two element lattice
$\left(\set{\bot,\top}; \meet, \join \right)$ in which $\bot < \top$.
This reference is done by passing from
$x\in\set{0,0_{\meet,l},0_{\meet,r},0_{\join,l},0_{\join,r}}$ to $\h x = \bot$ and for
$x\in\set{1,1_{\meet,l},1_{\meet,r},1_{\join,l},,1_{\join,r}}$ to $\h x = \top$
and putting
\begin{eqnarray*}
D_l(x,y,z) &=&
\left\{
\begin{array}{ll}
1_{\join,l}, & \mbox{\rm if $x,y,z\in\set{0,1}$ and $\h x\join \h y \join \h z = \top$,}\\
0_{\join,l}, & \mbox{\rm if $x,y,z\in\set{0,1}$ and $\h x\join \h y \join \h z = \bot$,}\\
s_{\join,l}, & \mbox{\rm otherwise,}
\end{array}
\right.\\
x \meet_l y&=&
\left\{
\begin{array}{ll}
1_{\meet,l}, & \mbox{\rm if $x,y\in\set{0_{\meet,l},1_{\meet,l},0_{\join,l},1_{\join,l}}$
                        and $\h x\meet \h y  = \top$,}\\
0_{\meet,l}, & \mbox{\rm if $x,y\in\set{0_{\meet,l},1_{\meet,l},0_{\join,l},1_{\join,l}}$
                        and $\h x\meet \h y  = \bot$,}\\
s_{\meet,l}, & \mbox{\rm otherwise,}
\end{array}
\right.\\
C_r(x,y,z) &=&
\left\{
\begin{array}{ll}
1_{\meet,r}, & \mbox{\rm if $x,y,z\in\set{0,1}$ and $\h x\meet \h y \meet \h z = \top$,}\\
0_{\meet,r}, & \mbox{\rm if $x,y,z\in\set{0,1}$ and $\h x\meet \h y \meet \h z = \bot$,}\\
s_{\meet,r}, & \mbox{\rm otherwise,}
\end{array}
\right.\\
x \join_r y&=&
\left\{
\begin{array}{ll}
1_{\join_r}, & \mbox{\rm if $x,y\in\set{0_{\meet,r},1_{\meet,r},0_{\join,r},1_{\join,r}}$
                        and $\h x\join \h y  = \top$,}\\
0_{\join_r}, & \mbox{\rm if $x,y\in\set{0_{\meet,r},1_{\meet,r},0_{\join,r},1_{\join,r}}$
                        and $\h x\join \h y  = \bot$,}\\
s_{\join_r}, & \mbox{\rm otherwise,}
\end{array}
\right.\\
f_l(x) &=&
\left\{
\begin{array}{ll}
1', & \mbox{\rm if $x=1_{\meet,l}$,}\\
x,  & \mbox{\rm otherwise,}
\end{array}
\right.\\
f_r(x) &=&
\left\{
\begin{array}{ll}
0', & \mbox{\rm if $x=0_{\join,r}$,}\\
x,  & \mbox{\rm otherwise.}
\end{array}
\right.\\
\end{eqnarray*}
A careful inspection of the above definitions shows that
there not many ways to compose operations of
$\m A = \left(A; D_l, C_r,\meet_l, \join_r, f_l, f_r \right)$
in a meaningful way, i.e. to get polynomials that have essential arity at least $4$.
Moreover for two such polynomials $\po t, \po s$ either they have disjoint ranges or the equation
$\po t(\o x)=\po s(\o x)$ has a solution.
However disjointness of the ranges can be checked by inspecting how the polynomials $\po t, \po s$
are built from the basic operations.
This shows that $\csat A \in P$.

\medskip
Now let $\Theta$ be the congruence with one nontrivial block $\set{0',1'}$.
This opens the way to transform the system of two lattice equations
\begin{eqnarray*}
(x^1_1\join x^1_2\join x^1_3) \meet \ldots \meet (x^m_1\join x^m_2\join x^m_3) &=& 1\\
(y^1_1\meet y^1_2\meet y^1_3) \join \ldots \join (y^n_1\meet y^n_2\meet y^n_3) &=& 0
\end{eqnarray*}
into a single equation
\[
f_l\left(D_l(x^1_1,x^1_2,x^1_3) \meet_l \ldots \meet_l D_l(x^m_1,x^m_2,x^m_3)\right) =
f_r\left(C_r(y^1_1,y^1_2,y^1_3) \join_r \ldots \join_r C_r(y^n_1,y^n_2,y^n_3)\right),
\]
so that the system is solvable in two element lattice if and only if this single equation is satisfied in $\m A/\Theta$ by the very same $\set{0,1}$-values for all the variables.
Together with Proposition \ref{prp-dl01} this allows us to conclude that
 $\csat {A/\theta}$ is \npc.
\end{proof}

\ssection{Type \tn 3 need not apply
\label{sect-type3}}

The most classical problem of solving equation is satisfiability of Boolean formulas
which is actually a paradigm for \npc \ problems.
The presence of Boolean behavior inside a finite algebra is in fact ruled out by the following theorem.

\begin{thm}
\label{thm-type-3}
If $\m A$ is finite algebra from a congruence modular variety
such that $\tn 3 \in \typset{\m A}$,
then $\cpolsatstar A$ is \npc.
\end{thm}

\begin{proof}
Suppose that $\m A$ is a finite algebra containing a type \tn 3 minimal set $U=\set{0,1}$ with respect to some covering pair $\alpha \prec \beta$ of its congruences.
Then $\rst{\m a}{U}$ is polynomially equivalent to a 2-element Boolean algebra, so that there are polynomials $\meet, \join, \neg$ of $\m A$ that behave on $U$ like meet, join and negation, respectively.
Moreover, there is a unary idempotent polynomial $\po e_{U}$ of $\m A$ with the range $U$.

Now the 3-SAT instance:
\[
\Phi \equiv \mmeet_{i=1}^m \ell^i_1 \join \ell^i_2 \join \ell^i_3,
\]
where $\ell^i_j \in \set{x^i_j, \neg x^i_j}$,
can be easily translated to the equation
\begin{eqnarray}
\label{bsat}
\mmeet_{i=1}^m \delta^i_1\po e_{U}(z^i_1) \join
               \delta^i_2\po e_{U}(z^i_2) \join
               \delta^i_3\po e_{U}(z^i_3) &=& 1,
\end{eqnarray}
where
\[
\delta^i_j\po e_{U}(z^i_j) =
\left\{
\begin{array}{ll}
\po e_{U}(z^i_j),      & \mbox{\rm if the literal $\ell^i_j$ is the variable, i.e., $\ell^i_j=x^i_j$},\\
\neg \po e_{U}(z^i_j), & \mbox{\rm if $\ell^i_j$ is the negated variable, i.e., $\ell^i_j=\neg x^i_j$}.
\end{array}
\right.
\]
It should be obvious that the formula $\Phi$ is satisfiable if and only if
the equation (\ref{bsat}) has a solution.
\end{proof}

Combining Theorems \ref{thm-cm-types} and \ref{thm-type-3} we get the following corollary.

\begin{cor}
\label{cor-typeset24}
If $\m A$ is finite algebra from a congruence modular variety
such that $\cpolsatstar A$ is not \npc 
then $\typset{\m A} \ci \set{\tn 2, \tn 4}$
\myqed
\end{cor}
\ssection{Transfer principles and decomposition
\label{sect-decomp}}

In this section we prove that every finite algebra $\m A$ from a congruence modular variety for which $\cpolsatstar A$ is not \npc \ decomposes into a direct product of a solvable algebra and an algebra that has only $\tn 4$ in its typeset.
In order to obtain such a nice decomposition we will first establish so called transfer principles
introduced by Matthew Valeriote in \cite{val-phd}.

\begin{df}
\label{df-ij-transfer}
We say that a finite algebra $\m A$ satisfies the $(\tn i, \tn j)$-transfer principle
if whenever $\alpha \prect{i} \beta \prect{j} \gamma$ are congruences of $\m A$
then there exists a congruence $\beta'$
with $\alpha \prect{j} \beta' \leq \gamma$.
\end{df}

The next Lemma helps us in a better localizing unwanted failures of the transfer principles.

\begin{lm}
\label{lm-ij-transfer}
If an algebra $\m A$ fails to have $(\tn i, \tn j)$-transfer principle
and $\con{\m A}$ is modular then
\begin{enumerate}
  \item \label{transfer-meet}
        $\m A$ has congruences $\alpha' \prect{i} \beta' \prect{j} \gamma'$
        with $\alpha'$ being meet irreducible,
  \item \label{transfer-join}
        $\m A$ has congruences $\alpha' \prect{i} \beta' \prect{j} \gamma'$
        with $\gamma'$ being join irreducible.
\end{enumerate}
\end{lm}

\begin{proof}
To see (\ref{transfer-meet}) suppose that the failure of $(\tn i, \tn j)$-transfer principle
is witnessed by the three element chain $\alpha \prect{i} \beta \prect{j} \gamma$.
Pick $\alpha'$ to be a maximal congruence that is over $\alpha$ but not over $\beta$.
Then obviously $\alpha'$ is meet irreducible,
as otherwise $\alpha'=\alpha_1 \cap \alpha_2$ with $\alpha_i > \alpha'$
would give $\alpha_i \geq \beta$ so that $\alpha'=\alpha_1 \cap \alpha_2 \geq \beta$.
One can easily check that
$\intv{\alpha}{\beta} \nearrow \intv{\alpha'}{\beta\join\alpha'}$
Moreover modularity of the lattice $\con{\m A}$ gives
$\intv{\beta}{\gamma} \nearrow \intv{\beta\join\alpha'}{\gamma\join\alpha'}$.
Summing up we get
$\alpha' \prect{i} \beta' \prect{j} \gamma'$
for $\beta' = \beta\join\alpha'$ and $\gamma' = \gamma\join\alpha'$.

The item (\ref{transfer-join}) can be shown in a dual way,
by replacing $\gamma$ with a minimal congruence $\gamma'$
that is below $\gamma$ but not below $\beta$.
\end{proof}

The next two Theorems establish both possible transfer principles,
as the typeset $\typset{\m A}$ is restricted in Corollary \ref{cor-typeset24}.

\begin{thm}
\label{thm-24transfer}
If $\m A$ is finite algebra from a congruence modular variety
in which $(\tn 2, \tn 4)$-transfer principle fails,
then $\cpolsatstar A$ is \npc.
\end{thm}

\begin{proof}
Suppose that $(\tn 2, \tn 4)$-transfer principle fails in $\m A$.
By Lemma \ref{lm-ij-transfer}.(\ref{transfer-join}) this failure can be witnessed with a
three element chain of congruences $\theta \prect{2} \alpha \prect{4} \beta$
with $\beta$ being a join irreducible.
Let $U=\set{0,1}$ be an $(\alpha,\beta)$-minimal set
and $V$ be an $(\theta,\alpha)$-minimal set in $\m A$.
Moreover let $\po e_U$ and $\po e_V$
be unary idempotent polynomials of $\m A$ with the range $U$ and $V$, respectively.
Taking into account the types of minimal sets $U$ and $V$ we know that
$\m A$ has the polynomials $\meet, \join$ that serve as the lattice operations on $\m A|_U$
(with respect to the lattice order $0<1$)
and a polynomial $\po d_V(x,y,z)$ that has the range contained in $V$
and is a Malcev operation on $V$.

\medskip

\begin{senumerate}
\item \label{24-projection}
For every $(a,b)\in \rst{\alpha}{V}$ there is a unary polynomial $\po f_{ab}(x)$ of $\m A$
such that $\po f_{ab}\vpair{0}{1} = \vpair{a}{b}$
\end{senumerate}
To produce such a polynomial $\po f_{ab}$ note that $(a,b) \in \rst{\alpha}{V} \ci \beta$
and $\beta = \cg{A}{0}{1}$, as $\beta$ is join irreducible and $(0,1)\not\in \alpha$.
Now simply recall Lemma \ref{lm-malcev-chain}.

Now with the help of (\ref{sect-decomp}.\ref{24-projection}) we will transform the system of two lattice equations
\begin{senumerate}
\item\label{lat-eqn}
$\displaystyle{
\left\{
\begin{array}{lcl}
  \mmeet_{i=1}^m x^i_1 \join x^i_2 \join x^i_3 &=& 1\\
  \jjoin_{i=1}^n y^i_1 \join y^i_2 \join y^i_3 &=& 0
\end{array}
\right.
}$
\end{senumerate}
into a single equation of the algebra $\m A$.
In view of Proposition \ref{prp-dl01} this will establish \npc{}ness of $\csat A$.

We start with picking $(a,b) \in \rst{\alpha}{V}-\theta$ to put
\begin{senumerate}
\item\label{24eqn}
$\po d_V \left( \po f_{ab}(\nu(\o y)), a, \po f_{ab}(\nu(\o y)\join \pi(\o x))\right) = b$,
where
\begin{eqnarray*}
  \pi(\o x) &=& \mmeet_{i=1}^m \po e_U(x^i_1) \join \po e_U(x^i_2) \join \po e_U(x^i_3), \\
  \nu(\o y) &=& \jjoin_{i=1}^n \po e_U(y^i_1) \join \po e_U(y^i_2) \join \po e_U(y^i_3).
\end{eqnarray*}
\end{senumerate}

First note that if $\o x, \o y$ is the $\set{0,1}$-lattice solution to (\ref{sect-decomp}.\ref{lat-eqn}) then keeping the values for the $x$'s and $y$'s we have
$\po f_{ab}(\nu(\o y))=a$ and $\po f_{ab}(\nu(\o y)\join \pi(\o x))=b$ so that
\[
\po d_V \left( \po f_{ab}(\nu(\o y)), a, \po f_{ab}(\nu(\o y)\join \pi(\o x))\right) =
\po d_V(a,a,b) = b,
\]
as required.

Conversely, if (\ref{sect-decomp}.\ref{24eqn}) has a solution $\o x, \o y$ in $\m A$ then
\begin{itemize}
  \item $\po f_{ab}(\nu(\o y)) = a = \po f_{ab}(\nu(\o y)\join \pi(\o x))$
        is not possible, \\
        as $\po d_V(a,a,a)=a$,
  \item $\po f_{ab}(\nu(\o y)) = b$ and $\po f_{ab}(\nu(\o y)\join \pi(\o x))=a$
        is not possible, \\
        as then we would have $1 = \nu(\o y) \leq \nu(\o y) \join \pi(\o x) = 0$ in the set $U$, contrary to our choice of $0<1$,
  \item $\po f_{ab}(\nu(\o y)) = b = \po f_{ab}(\nu(\o y)\join \pi(\o x))$
        is not possible, \\
        as then $\po d_V(b,a,b)= b$ contrary to the fact that
        $\po d_V(b,b,b)=b$ and $x \mapsto \po d_V(b,x,b)$ is a permutation of $V$.
\end{itemize}
Thus the only possibility for a solution $\o x, \o y$ to (\ref{sect-decomp}.\ref{24eqn}),
is to satisfy $\po f_{ab}(\nu(\o y))=a$ and $\po f_{ab}(\nu(\o y)\join \pi(\o x))=b$,
or equivalently that $\nu(\o y)=0$ and $\pi(\o x)=1$.
Therefore evaluating the $x^i_j$'s and $y^i_j$'s by $\po e_U(x^i_j)$ and $\po e_U(y^i_j)$, respectively, we get a solution to the system of lattice equations (\ref{sect-decomp}.\ref{lat-eqn}).
\end{proof}

\begin{thm}
\label{thm-42transfer}
If $\m A$ is finite algebra from a congruence modular variety
in which $(\tn 4, \tn 2)$-transfer principle fails,
then $\cpolsatstar A$ is \npc.
\end{thm}

\begin{proof}
Suppose that $(\tn 4, \tn 2)$-transfer principle fails in $\m A$.
By Lemma \ref{lm-ij-transfer}.(\ref{transfer-join}) this failure can be witnessed with a
three element chain $\theta \prect{4} \alpha \prect{2} \beta$ with $\theta$ being a meet irreducible congruence.
Let $U$ be an $(\alpha,\beta)$-minimal set
and $V=\set{0,1}$ be an $(\theta,\alpha)$-minimal set in $\m A$.
Moreover let $\po e_U$ and $\po e_V$
be unary idempotent polynomials of $\m A$ with the range $U$ and $V$, respectively.
Taking into account the types of minimal sets $U$ and $V$ we know that
$\m A$ has the polynomials $\meet, \join$ that serve as the lattice operations on $\m A|_V$
(with respect to the lattice order $0<1$)
and a polynomial $\po d_U(x,y,z)$ that has the range contained in $U$
and is a Malcev operation on $U$.

We start with the following claims:
\begin{senumerate}
\item\label{01incgcd}
    $(0,1) \in \cg{A}{c}{d}$ for each $(c,d)\not\in \theta$,
\item \label{42-projection}
    for every $(c,d)\in U^2 - \theta$ there is a unary polynomial $\po f_{cd}(x)$ such that
    $\po f_{cd}\vpair{c}{d} = \vpair{0}{1}$,
\item \label{alphaUintheta}
    $\alpha|_U \ci \theta$.
\end{senumerate}
Note that $(0,1) \in \alpha \leq \theta \join \cg{A}{c}{d}$
which together with Lemma \ref{lm-malcev-chain} gives a Malcev chain that connects $0$ with $1$
via projections by unary polynomials pairs from $\theta$ or the pair $(c,d)$.
Applying $\po e_V$ to this chain we get that it entirely lives in $\set{0,1}$.
Since $\theta|_V =0$ we get that at some link in this chain the pair ${0,1}$ occurs as a projection by a unary polynomial $\po f$ applied to the pair $(c,d)$.
This obviously gives $(0,1) \in \cg{A}{c}{d}$.

If, as in the assumptions of (\ref{sect-decomp}.\ref{42-projection}), $(c,d)\in U^2 - \theta$
then either $\po f$ is good enough to serve as $\po f_{cd}$
or we put $\po f_{cd}(x) = \po f\po d_U(c,x,d)$.

To see (\ref{sect-decomp}.\ref{alphaUintheta})
suppose to the contrary that $(0',1') \in \alpha|_U-\theta$.
In particular $(0',1') \in \alpha = \theta \join \cg{A}{0}{1}$.
Thus there is a Malcev chain connecting $0'$ and $1'$
via links of the form $\set{\po f(c), \po f(d)}$, with $\po f$ being unary polynomials of $\m A$
and $(c,d)\in \theta \cup \set{0,1}$.
By applying $\po e_U$ to this chain we may assume that it is fully contained in $U$.
Since $(0',1') \not\in \theta$ at least one link must be obtained by projecting the set $\set{0,1}$ onto a pair $(0'',1'')\in\alpha|_U-\theta$.
Therefore $\set{0'',1''}$ is a $(\theta,\alpha)$-minimal set of type $\tn 4$ lying inside a minimal set $U$ of type $\tn 2$, which is not possible.

\bigskip

Now we pick a transversal $\set{c_0,c_1,\ldots, c_k}$ of $U/\alpha$ and define a unary polynomial
\[
\po s(x) = \jjoin_{i=1}^k \po f_{c_0c_i}(x)
\]
so that $\po s(c_0)=0$ and $\po s(c_i) = 1$ for $i=1,\ldots,k$.
In fact
\begin{senumerate}
\item
for every $a \in U$ we have
$\displaystyle{
\po s(a) =
\left\{
\begin{array}{ll}
0, & \mbox{\rm if $a \in c_0/\alpha$,}\\
1, & \mbox{\rm otherwise,}
\end{array}
\right.
}$
\end{senumerate}
as for every $a \in U \cap c_i/\alpha$ we have $(\po s(a), \po s(c_i))\in \alpha|_U \ci \theta$,
which together with $\theta|_\set{0,1}=0$ gives $\po s(a)=\po s(c_i)$.

\medskip

Now our proof of \npc ss of $\cpolsatstar A$ splits into two cases depending on the size of $U/\alpha$.

\bigskip

\noindent
\textsc{Case 1.} $\card{U/\alpha}=2$, i.e., $U=\set{c_0,c_1}$.

\medskip
In this case we will transform each 3-SAT instance:
\[
\Phi \equiv \mmeet_{i=1}^m \ell^i_1 \join \ell^i_2 \join \ell^i_3,
\]
where $\ell^i_j \in \set{x^i_j, \neg x^i_j}$,
to an equation
\begin{senumerate}
\item\label{42eqn}
$\displaystyle{
\mmeet_{i=1}^m z^i_1 \join z^i_2 \join z^i_3 = 1
}$,
where
\[
z^i_j =
\left\{
\begin{array}{ll}
\po s \po e_U(x^i_j),                  & \mbox{\rm if $\ell^i_j = x^i_j$},\\
\po s \po d_U(\po e_U(x^i_j),c_0,c_1), & \mbox{\rm if $\ell^i_j = \neg x^i_j$}.
\end{array}
\right.
\]
\end{senumerate}

It should be obvious that each evaluation satisfying $\Phi$ can be transformed into a solution of (\ref{sect-decomp}.\ref{42eqn}) by sending the $0$'s to $c_0$ and the $1$'s to $c_1$.

Conversely, if the $a^i_j$'s form a solution to (\ref{sect-decomp}.\ref{42eqn}) then putting
\[
x^i_j =
\left\{
\begin{array}{ll}
0, & \mbox{\rm if $\po e_U(a^i_j) \in c_0/\alpha \cap U$,}\\
1, & \mbox{\rm if $\po e_U(a^i_j) \in c_1/\alpha \cap U$,}
\end{array}
\right.
\]
we get a valuation satisfying $\Phi$.

\bigskip

\noindent
\textsc{Case 2.} $\card{U/\alpha} \geq 3$.

\medskip
In this case with each graph $G=(V,E)$ with $V=\set{v_1,\ldots,v_n}$ we associate a polynomial
$\po t_G(x_1,\ldots,x_n) \in \pol{A}$ in such a way that $G$ is $\card{U/\alpha}$-colorable
iff the equation $\po t_G(\o x)=1$ has a solution in $\m A$.

First observe that for the polynomial
$\po s'(x,y) = \po s \po d_U(x,y,c_0)$ and $a,b\in U$ we have
\begin{senumerate}
\item\label{alpha-class}
$\displaystyle{
\po s'(a,b) =
\left\{
\begin{array}{ll}
0, & \mbox{\rm if $(a,b)\in \alpha$,} \\
1, & \mbox{\rm otherwise.}
\end{array}
\right.
}$
\end{senumerate}
This is due to the fact that the polynomial $\po d_U(x,y,z)$, after fixing values of any two of its variables (by values in $U$), is a permutation of $U$ with respect to the remaining variable.

\medskip
Now we put
\[
\po t_G(x_1,\ldots,x_n) = \mmeet_{\set{v_i,v_j}\in E} \po s'(\po e_U(x_i), \po e_U(x_j)).
\]
From (\ref{sect-decomp}.\ref{alpha-class}) we know that $\po t_G(x_1,\ldots,x_n) = 1$ iff
for all $i,j$ such that $v_i$ and $v_j$ are connected in $G$ by an edge,
$\po e_U(x_i)$ and $\po e_U(x_j)$ are in different $\alpha|_U$-classes.
This means that $\po t_G(\o x) = 1$ has a solution in $\m A$ iff $G$ can be properly colored by
$\alpha|_U$-classes.
\end{proof}

With the help of Theorems \ref{thm-24transfer} and \ref{thm-42transfer} we are ready to prove the promised decomposition.

\begin{cor}
\label{cor-decomposition}
If $\m A$ is finite algebra from a congruence modular variety
then either $\m A$ is isomorphic to a direct product $\m A_2 \times \m A_4$,
where $\typset{\m A_2} \ci \set{\tn 2}$ and $\typset{\m A_4} \ci \set{\tn 4}$
or $\cpolsatstar A$ is \npc.
\end{cor}

\begin{proof}
Suppose that $\cpolsatstar A$ is not \npc.
From Theorem \ref{thm-type-3} we know that $\typset{A} \ci \set{\tn 2, \tn 4}$.
To get the required decomposition we start with the following easy claim:

\medskip

\begin{senumerate}
\item \label{modular-iso}
For $\alpha,\beta \in \cn{A}$ we have
$\typset{\alpha \cap \beta, \alpha} = \typset{\beta, \alpha \join \beta}$.
\end{senumerate}
In a modular lattice the intervals
$\intv{\alpha \cap \beta}{\alpha}$ and
$\intv{\beta}{\alpha \join \beta}$ are isomorphic
under the mutually converse mappings
$\gamma \mapsto \gamma\join\beta$ and
$\alpha \cap \delta \tomaps \delta$.
In particular every covering pair $\alpha \cap \beta \leq \gamma \prec \gamma' \leq \alpha$
is mapped onto a covering pair
$\beta \leq \gamma\join\beta \prec \gamma'\join\beta \leq \alpha\join\beta$.
In fact $\intv{\gamma}{\gamma'}\nearrow \intv{\gamma\join\beta}{\gamma'\join\beta}$
so that $\typ({\gamma},{\gamma'})=\typ({\gamma\join\beta},{\gamma'\join\beta})$.

\medskip

\begin{senumerate}
\item \label{radicals}
For $i \in \set{2,4}$ there is the largest congruence $\rho_i \in \cn{A}$
with $\typset{0,\rho_i} \ci \set{\tn i}$.
\end{senumerate}
To prove this it suffices to show that for $\alpha,\beta \in \cn{A}$ with
$\typset{0,\alpha} \ci \set{\tn i} \co \typset{0,\beta}$ we have
$\typset{0, \alpha\join\beta} \ci \set{\tn i}$.
Thus let $\alpha \cap \beta \leq \delta \prec \delta' \leq \alpha \join \beta$.
If $\alpha\cap\delta < \alpha\cap\delta'$ or $\alpha\join\delta < \alpha\join\delta'$
that such an inequality is actually a covering so that, by our assumptions and (\ref{sect-decomp}.\ref{modular-iso}), it has type $\tn i$.
Therefore $\delta \prec \delta'$ inherits type either from
$\alpha\cap\delta \prect{i} \alpha\cap\delta'$ or from
$\alpha\join\delta \prect{i} \alpha\join\delta'$.
On the other hand at least one of those strong inequalities has to hold, 
as otherwise the congruences $\alpha\cap\delta \leq \alpha, \delta, \delta' \leq \alpha\join\delta$
would form a pentagon, contradicting the modularity of $\cn{A}$.

\medskip

\begin{senumerate}
\item \label{mod-radicals}
For $\set{i,j} = \set{2,4}$ we have $\typset{\rho_i,1} \ci \set{\tn j}$.
\end{senumerate}
Suppose that (\ref{sect-decomp}.\ref{mod-radicals}) fails,
and $\alpha$ is a minimal congruence above $\rho_i$ that has a cover, say $\beta$, of type $\tn i$.
Since, by the definition of $\rho_i$ all its covers are of type $\tn j$
we know that $\rho < \alpha \prect{i} \beta$.
Therefore there is $\theta\in \cn{A}$ with $\rho \leq \theta \prect{j} \alpha \prect{i} \beta$.
By Theorems \ref{thm-24transfer} and \ref{thm-42transfer} we know that there is $\theta'\in \cn{A}$
with $\theta \prect{i} \theta' \leq \beta$.
This contradicts the minimality of $\alpha$ as $\theta < \alpha$
and $\theta$ has a cover of type $\tn i$.

\medskip

From (\ref{sect-decomp}.\ref{mod-radicals}) we know that for
$\m A_2 = \m A/\rho_4$ and $\m A_4 = \m A/\rho_2$ we have $\typset{\m A_i} \ci \set{\tn i}$.
To show that $\m a$ is isomorphic with the product $\m A_2 \times \m A_4$
first note that $\rho_2 \cap \rho_4 = 0$, by the definitions of $\rho_i$,
and $\rho_2 \join \rho_4 = 1$, by (\ref{sect-decomp}.\ref{mod-radicals}).
Finally, by Theorem 6.2 of \cite{fm}, $\rho_2$ permutes with all congruences of $\m A$, as $\rho_2$ is solvable. This gives that $\rho_2, \rho_4$ gives a factorization of $\m A$, as required.
\end{proof}

The decomposition established in Corollary \ref{cor-decomposition} 
together with the possibility of passing to the quotients 
allows us to separately consider solvable algebras,
i.e. algebras with typeset contained in $\set{\tn 2}$ and
entirely lattice type algebras, i.e. algebras with typeset contained in $\set{\tn 4}$. 
\ssection{Restricting solvable behavior
\label{sect-solvable}}

The aim of this section is to show that every finite solvable algebra $\m A$ from a congruence modular variety is in fact nilpotent or $\m A$ has a homomorphic image $\m A^\prime$ with
$\cpolsatstar {A^\prime}$ being \npc. We start with the following construction.

\begin{lm}
\label{lm-solvable}
Let $\m A$ be a finite solvable subdirectly irreducible algebra from a congruence modular variety.
If $\comm{1}{\mu}>0$, where $\mu$ is the monolith of $\m A$
then $\cpolsatstar A$ is \npc.
\end{lm}

\begin{proof}
Put $\alpha =\centr{\mu}{0}$.
If $\comm{1}{\mu}>0$ then there is $\beta \in \cn{A}$ such that
$\mu \leq \alpha \prec \beta$
and obviously $\alpha \prect{2} \beta$ by the solvability of $\m A$.
Moreover we have $\comm{\alpha}{\mu}=0$ while $\comm{\beta}{\mu}=\mu$.
Pick:
\begin{itemize}
  \item an $(\alpha,\beta)$-minimal set $U$,
  \item a transversal $\set{d_0,d_1,\ldots,d_k}$ of $U/\alpha$,
  \item a $(0,\mu)$-minimal set $V$,
  \item a pair $(e,a) \in \mu|_V-0$ and let $N=e/\mu \cap V$ be the trace of $V$ containing both $e$ and $a$.
\end{itemize}
We know that $\m A|_N$ is polynomially equivalent to a (one dimensional) vector space and we may assume that $e$ is its zero element with respect to the vectors addition $+$ which has to be a polynomial of $\m A$.

Now note that by the choice of the $d_i$'s
we know that $\alpha < \alpha \join \cg{A}{d_i}{d_j}$ for $i \neq j$
which gives $\comm{\cg{A}{d_i}{d_j}}{\mu}=\mu$.
This has to be witnessed by a polynomial $\po s_{ij}(x,y_1,\ldots,y_m)$ and elements
$(c,d)\in \cg{A}{d_i}{d_j}$ and $(a_1,b_1),\ldots, (a_m,b_m) \in \mu = \cg{A}{e}{a}$
satisfying
\begin{eqnarray*}
  s'_{ij}(c,a_1,\ldots,a_m) & =    & s'_{ij}(c,b_1,\ldots,b_m), \\
  s'_{ij}(d,a_1,\ldots,a_m) & \neq & s'_{ij}(d,b_1,\ldots,b_m).
\end{eqnarray*}
Since $\m A$ is a solvable algebra in a congruence modular variety, the variety generated by $\m A$ is solvable and therefore congruence permutable (by Theorem 6.3 of \cite{fm}).
Thus $\m A$ has a Malcev term, say $\po d$,
and by Lemma \ref{lm-malcev-chain} we get unary polynomials
$\po q, \po p_1,\ldots,\po p_m$ of $\m A$ with
$\po q(d_i)=c, \po q(d_j)=d$ and $\po p_k(e)=a_k, \po p_k(a)=b_k$ for all $1\leq k\leq m$.
Thus for the polynomial
$\po s_{ij}(x,y) = \po e_V\po s'_{ij}(\po q(y),\po p_1(x),\ldots,\po p_m(x))$
we have
\begin{eqnarray*}
  \po s_{ij}(e,d_i) & =     & \po s_{ij}(a,d_i), \\
  \po s_{ij}(e,d_j) & \neq  & \po s_{ij}(a,d_j).
\end{eqnarray*}
Again referring to Lemma \ref{lm-malcev-chain}
and using  $(e,a) \in \cg{A}{\po s_{ij}(e,d_j)}{\po s_{ij}(a,d_j)}$ we get a unary polynomial
$\po p$ of $\m A$ that takes the pair $({\po s_{ij}(e,d_j)},{\po s_{ij}(a,d_j)})$ to $(e,a)$.
Now, replacing $\po s_{ij}(x,y)$ by $\po d(\po s_{ij}(x,y), \po s_{ij}(e,y), e)$
we get that
\[
\begin{array}{rcccccl}
     && \po s_{ij}(e,d_i) & =     & \po s_{ij}(a,d_i) & = & e,\\
e & = & \po s_{ij}(e,d_j) & \neq  & \po s_{ij}(a,d_j) & = & a.
\end{array}
\]
In fact we know that $\po s_{ij}(e,y)=e$ for all $y\in A$.

Now, for each fixed $y \in A$ the unary polynomial
\[
V \ni v \mapsto \po s_{ij}(v,y) \in V
\]
is either a permutation of $V$ or collapses $\mu|_V$ to $0$,
i.e., it is constant on $\mu|_V$-classes.
Thus, iterating $\po s_{ij}(v,y)$ in the first variable a sufficient number of times
we can modify $\po s_{ij}$ to additionally satisfy that for each fixed $y \in A$ the new polynomial $\po s_{ij}(v,y)$ is either the identity map on $V$ or it is constant on $\mu|_V$-classes.
Actually, in the second case, i.e. if $\po s_{ij}(v,y)$ collapses $\mu|_V$ to $0$
then it collapses the trace $N$ to $\po s_{ij}(e,y)=e$.
Summing up, we produced polynomials $\po s_{ij}$ satisfying
\[
\begin{array}{rccl}
\po s_{ij}(e,y)   &=& e, &\mbox{\rm for each \ } y\in A,\\
\po s_{ij}(v,d_i) &=& e, &\mbox{\rm for each \ } v\in N,\\
\po s_{ij}(v,d_j) &=& v, &\mbox{\rm for each \ } v\in V.
\end{array}
\]
Now, using the fact that $\comm{\mu}{\alpha}=0$
we can keep the above equalities by varying the second variable modulo $\alpha$
\[
\begin{array}{rccl}
\po s_{ij}(e,y) &=& e, &\mbox{\rm for each \ } y\in A,\\
\po s_{ij}(v,y) &=& e, &\mbox{\rm for each \ } v\in N \mbox{\rm \ and \ } y \in d_i/\alpha,\\
\po s_{ij}(v,y) &=& v, &\mbox{\rm for each \ } v\in V \mbox{\rm \ and \ } y \in d_j/\alpha.
\end{array}
\]
Now for $j=0,\ldots,k$ define
\[
\po s_j(x,y) = \po s_{i_1j}(\ldots \po s_{i_{k-1}j}(\po s_{i_{k}j}(x,y),y)\ldots,y),
\]
where $\set{j,i_1,\ldots,i_k}=\set{0,1,\ldots,k}$.
It is easy to observe that $\po s_j$ has the range contained in $V$ and
\[
\begin{array}{cccl}
\po s_{j}(e,y) &=& e, &\mbox{\rm for each \ } y\in A,\\
\po s_{j}(v,y) &=& e, &\mbox{\rm for each \ } v\in N \mbox{\rm \ and \ } y \not\in d_j/\alpha,\\
\po s_{j}(v,y) &=& v, &\mbox{\rm for each \ } v\in V \mbox{\rm \ and \ } y \in d_j/\alpha.
\end{array}
\]
Indeed, the first and the last item follows directly from the definition of $\po s_j$.
To see the middle one note that for $v\in N$ and $y \in d_{i_\ell}/\alpha$ we have
\begin{eqnarray*}
v' & =             & \po s_{i_{\ell+1}}(\ldots \po s_{i_{k-1}j}(\po s_{i_{k}j}(v,y),y)\ldots,y)\\
   &\congruent{\mu}& \po s_{i_{\ell+1}}(\ldots \po s_{i_{k-1}j}(\po s_{i_{k}j}(e,y),y)\ldots,y)\\
   & = & e,
\end{eqnarray*}
i.e. $v'\in N$ so that $\po s_{i_\ell j}(v',y)=e$, and consequently
\begin{eqnarray*}
\po s_{j}(v,y) &=& \po s_{i_1j}(\ldots \po s_{i_{\ell-1}j}(\po s_{i_{\ell}j}(v',y),y)\ldots,y)\\
               &=& \po s_{i_1j}(\ldots \po s_{i_{\ell-1}j}(e,y)\ldots,y)\\
               &=& e.
\end{eqnarray*}

As $\m A|_N$ is polynomially equivalent to a vector space (with $e$ being its neutral element)
and for $v\in N$ and $y \in U$ the elements $\po s_j(v,y)$ are in $\m A|_N$
then it makes sense to sum them up and define
\[
\po s(x,y) = \sum_{j=1}^k \po s_j(x,y)
\]
to get
\[
\begin{array}{rccl}
\po s(e,y) &=& e, &\mbox{\rm for each \ } y\in A,\\
\po s(v,y) &=& e, &\mbox{\rm for each \ } v\in N \mbox{\rm \ and \ } y \in U \cap d_0/\alpha,\\
\po s(v,y) &=& v, &\mbox{\rm for each \ } v\in V \mbox{\rm \ and \ } y \in U -    d_0/\alpha.
\end{array}
\]
Indeed, in the sum defining $\po s$, at most one summand differs from $e$,
namely $\po s_j(v,y)$ for the unique $j$ such that $y \in U \cap d_j/\alpha$.

Now we are ready to code each instance of an \npc \ problem in a single equation of $\m A$ endowed with some polynomials.
As for Theorem \ref{thm-42transfer}
our proof splits into two cases depending on the size of $U/\alpha$.

\bigskip

\noindent
\textsc{Case 1.} $\card{U/\alpha}\geq 3$.

\medskip
In this case with each graph $G=(V,E)$ with $V=\set{v_1,\ldots,v_n}$ we associate a polynomial
$\po t_G(x_1,\ldots,x_n) \in \poln n {A}$ in such a way that $G$ is $\card{U/\alpha}$-colorable
iff the equation $\po t_G(\o x)=1$ has a solution in $\m A$.

For more readability we define polynomials of the form  $v \with_{\po s} Y$
acting on $N \times U^m$ as follows
\[
v \with_{\po s} \set{y_1,\ldots,y_m} = \po s(\ldots \po s(\po s(v,y_1),y_2)\ldots,y_m).
\]
Note that if $(v,\o y) \in N \times U^m$ then the value of $v \with_{\po s} Y$ does not depend on the order of the $y_i$'s and in fact we have
\[
a \with_{\po s} Y =
\left\{
\begin{array}{ll}
a, & \mbox{\ if \ } Y \cap d_0/\alpha = \emptyset,\\
e, & \mbox{\ otherwise.}
\end{array}
\right.
\]
Moreover note that for the polynomial $x-y = \po d_U(x,y,d_0)$ and $x,y\in U$ we have
\[
x-y\in d_0/\alpha \mbox{\ iff \ } (x,y)\in \alpha,
\]
as $\po d_U(x,y,d_0)$ is a permutation of $U$
whenever the value for one of the variables $x,y$ is fixed.

Now for a graph $G$ define $\po t_G(\o x)$ by putting
\[
t_G(x_1,\ldots,x_n) = a \with_{\po s}\set{\po e_U(x_i) - \po e_U(x_j) : \set{v_i,v_j}\in E}.
\]
From what it was said about $a \with_{\po s} Y$ and the difference $-$ on $\m A|_N$
it should be clear that the equation $\po t_G(\o x)=a$ has a solution in $\m A$
iff the elements $\po e_U(x_i)$ and $\po e_U(x_j)$, corresponding to the edge $\set{v_i,v_j}$,
are evaluated in different $\alpha|_U$-classes,
i.e., if $G$ is $(k+1)$-colorable.

\bigskip

\noindent
\textsc{Case 2.} $\card{U/\alpha}=2$, i.e., $U=\set{d_0,d_1}$.

\medskip

Being in this case we start with the following polynomial $\po w(v,y_1,y_2,y_3)$
of $\m A$ acting on $N \times U^3$ as follows
\begin{eqnarray*}
\po w(v,y_1,y_2,y_3) &=& \po s(\po s(\po s(v,y_1),y_2),y_3) \\
                     &&  -\po s(\po s(v,y_1),y_2)-\po s(\po s(v,y_1),y_3)-\po s(\po s(v,y_2),y_3)\\
                     &&  +\po s(v,y_1)+\po s(v,y_2)+\po s(v,y_3),
\end{eqnarray*}
where the addition + and the substraction $-$ is taken in the vector space $\m A|_N$.
One can easily check that
\[
\begin{array}{lcll}
\po w(e,y_1,y_2,y_3) &=& e &\mbox{for all $y_1,y_2,y_3\in A$},\\
\po w(v,y_1,y_2,y_3) &=& e &\mbox{for $v\in N$ and $\set{y_1,y_2,y_3} \ci U \cap d_0/\alpha$,}\\
\po w(v,y_1,y_2,y_3) &=& v &\mbox{for $v\in N$ and
                                                $\set{y_1,y_2,y_3}\cap(U-d_0/\alpha)\neq\emptyset$.}
\end{array}
\]
Analogously to Case 1 we define a polynomials of the form
$v \bigstar_{\po w} T$ acting on $N \times U^{3m}$,
where now $T =\set{(y^i_1,y^i_2,y^i_3) : i=1,\ldots,m}$ is a set of triples of variables,
by putting
\[
v \bigstar_{\po w} T
= \po w(\ldots(\po w(\po w(v,y^1_1,y^1_2,y^1_3),y^2_1,y^2_2,y^2_3)\ldots,y^m_1,y^m_2,y^m_3)).
\]
Again, if the variable $v$ is evaluated in $N$ and all the $y^i_j$'s in $U$
then the value of $v \bigstar_{\po w} T$ does not depend on the order of triples in $T$
neither on the order inside the triples. In fact we have
\[
a \bigstar_{\po w} T =
\left\{
\begin{array}{ll}
e, & \mbox{if there is $j=1,\ldots,m$ with $\set{y_1,y_2,y_3} \ci U \cap d_0/\alpha$,}\\
a, & \mbox{otherwise,}
\end{array}
\right.
\]
i.e., $a \bigstar_{\po w} T$ acts like a conjunction of disjunction of triples.
Indeed, a 3-SAT instance
\[
\Phi \equiv \mmeet_{i=1}^m \ell^i_1 \join \ell^i_2 \join \ell^i_3,
\]
where $\ell^i_j \in \set{x^i_j, \neg x^i_j}$,
can be translated to a polynomial
\[
\po t_\Phi(\o x) = a \bigstar_{\po w} \set{\set{z^i_1,z^i_2,z^i_3} : i=1,\ldots,m},
\]
where
\[
z^i_j =
\left\{
\begin{array}{ll}
\po e_U(x^i_j),
        & \mbox{\rm if the literal $\ell^i_j$ is the variable, i.e., $\ell^i_j=x^i_j$},\\
\po d_U(d_1,\po e_U(x^i_j), d_0),
        & \mbox{\rm if $\ell^i_j$ is the negated variable, i.e., $\ell^i_j=\neg x^i_j$}.
\end{array}
\right.
\]
First note that for $z^i_j = \po d_U(d_1,\po e_U(x^i_j), d_0)$
we have $z^i_j \in U \cap d_{1-\ell}/\alpha$ whenever $e_U(x^i_j) \in U \cap d_{\ell}/\alpha$,
i.e. $\po e_U(x) \mapsto \po d_U(d_1,\po e_U(x), d_0)$ acts as a negation on the set
$\set{d_0/\alpha, d_1/\alpha}$.
Moreover, from what has been already said about $a \bigstar_{\po w} T$,
it should be clear that the equation $\po t_\Phi(\o x) = a$ has a solution in $\m A$
iff $\Phi$ is satisfiable.
Indeed, it suffices to evaluate the $x$'s in $\Phi$ by the boolean value $\ell$ iff
in the corresponding solution of $\po t_\Phi(\o x) = a$ they are evaluated in a way that
$\po e_U(x) \in d_\ell/\alpha$.
\end{proof}

\begin{cor}
\label{cor-solvable}
If a finite algebra $\m A$ from a congruence modular variety is solvable but not nilpotent
then $\m A$ has a homomorphic image $\m A^\prime$ with $\cpolsatstar {A^\prime}$ being \npc.
\end{cor}

\begin{proof}
If $\m A$ is solvable but not nilpotent then there is a natural number $k$ such that
\[
1 > 1^{(2)} > \ldots > 1^{(k)} = 1^{(k+1)} > 0.
\]
Now, picking a maximal congruence $\fj$ which is not above $1^{(k)}$
we know that $\fj$ is meet-irreducible (with the unique cover $\fj^+$)
and that the quotient $\m A^\prime = \m a/\fj$ is solvable but not nilpotent,
as in $\m A^\prime$ we have
\[
1^{(k)} = 1^{(k+1)} = \fj^+/\fj.
\]
Now we are in a position to apply Lemma \ref{lm-solvable}.
\end{proof} 
\ssection{Restricting lattice behavior
\label{sect-lattices}}

In this section we study finite algebras from congruence modular varieties such that all prime quotients of its congruences are of lattice type, i.e., of type \tn 4.
We will show that if such an algebra $\m A$ is not a subdirect product of algebras each of which is polynomially equivalent to the 2-element lattice, then $\m A$
has a homomorphic image $\m A^\prime$ with $\polsatstar{A^\prime}$ being \npc.

\bigskip

In the first Lemma of this section we collect some configurations
that lead to \npc{}ness in type $\tn 4$ algebras.

\begin{lm}
\label{lm-pseudonegation}
Let $\m A$ be a finite algebra from a congruence modular variety such that
$\typset{\m A} \ci \set{\tn 4}$.
If one of the following configuration
\begin{enumerate}
  \item \label{lm-pseudonegation-1}
  $\set{0,1}$ is a minimal set and there are three different elements $a,b,c \in A$ and
  $\po f_a, \po f_c \in \poln 1 A$ such that
  $\po f_a\vpair{1}{0} = \vpair{b}{a}$ and
  $\po f_c\vpair{1}{0} = \vpair{c}{b}$,
  \item \label{lm-pseudonegation-2}
  $\set{0,1}$ is the range of a polynomial $\po p \in \poln 1 A$,
  $\set{a,b}$ is a minimal set
  and there are polynomials $\po f, \po g \in \poln 1 A$ such that
  $\po f\vpair{1}{0} = \vpair{b}{a}$ and
  $\po g\vpair{1}{0} = \vpair{a}{b}$,
  \item \label{lm-pseudonegation-3}
  $\set{0,1}$ is a minimal set and one of the sets
  \begin{eqnarray*}
  F_1 &=&
    \bigcap\set{\po f^{-1}(1) \ : \ \po f\in \poln 1 A \mbox{ \ and \ } \po f(A)=\set{0,1}},\\
  F_0 &=&
    \bigcap\set{\po f^{-1}(0) \ : \ \po f\in \poln 1 A \mbox{ \ and \ } \po f(A)=\set{0,1}},
  \end{eqnarray*}
  is empty
\end{enumerate}
can be found in $\m A$ then $\polsatstar{A}$ is \npc.
\end{lm}

\begin{proof}
Suppose we have a configuration described in (\ref{lm-pseudonegation-1}) 
and let $\po e$ be a unary idempotent polynomial with the range $\set{0,1}$.
In this case the required \npc{}ness follows from Proposition \ref{prp-dl01} by transforming
the system of the two lattice equations:
\begin{eqnarray*}
  \mmeet_{i=1}^m x^i_1 \join x^i_2 \join x^i_3 &=& 1,\\
  \jjoin_{i=1}^n y^i_1 \join y^i_2 \join y^i_3 &=& 0
\end{eqnarray*}
into a single equation
\[
\po f_a\left(\mmeet_{i=1}^m \po e(x^i_1) \join \po e(x^i_2) \join \po e(x^i_3)\right)
=
\po f_c\left(\jjoin_{i=1}^n \po e(y^i_1) \join \po e(y^i_2) \join \po e(y^i_3)\right)
\]
of the algebra $\m A$, where meets and joins are performed in the minimal set $\set{0<1}$.

In case (\ref{lm-pseudonegation-2}) we code the 3-SAT instance:
\[
\Phi \equiv \mmeet_{i=1}^m \ell^i_1 \join \ell^i_2 \join \ell^i_3,
\]
by the equation
\begin{eqnarray*}
\mmeet_{i=1}^m \delta^i_1\po p(z^i_1) \join
               \delta^i_2\po p(z^i_2) \join
               \delta^i_3\po p(z^i_3) &=& b,
\end{eqnarray*}
where
\[
\delta^i_j\po p(z^i_j) =
\left\{
\begin{array}{ll}
\po {fp}(z^i_j),
    & \mbox{\rm if the literal $\ell^i_j$ is the variable, i.e., $\ell^i_j=x^i_j$},\\
\po {gp}(z^i_j),
    & \mbox{\rm if $\ell^i_j$ is the negated variable, i.e., $\ell^i_j=\neg x^i_j$}.
\end{array}
\right.
\]
and meets and joins are taken in the minimal set $\set{a<b}$.

Finally, in case (\ref{lm-pseudonegation-3}) suppose that $F_1 = \emptyset$.
This in particular means that the set

\[
P=\set{\po f\in \poln 1 A \mbox{ \ and \ } \po f(A)=\set{0,1}}
\]
contains at least two different polynomials.
Note that if $\po g_1, \po g_2 \in P$ then for $\po g(x) = \po g_1(x) \meet \po g_2(x)$
we have $\po g^{-1}(1) = \po g_1^{-1}(1) \cap \po g_2^{-1}(1)$, and if this intersection is nonempty then also $\po g \in P$ as then $\po g(A)=\set{0,1}$.
Thus $F_1 = \emptyset$ gives that there are $\po f, \po g \in P$ with
$\po f^{-1}(1) \cap \po g^{-1}(1) = \emptyset$.
Now we can transform the 3-SAT instance $\Phi$ into the equation
\begin{eqnarray*}
\mmeet_{i=1}^m \delta^i_1(z^i_1) \join
               \delta^i_2(z^i_2) \join
               \delta^i_3(z^i_3) &=& 1
\end{eqnarray*}
where
\[
\delta^i_j(z^i_j) =
\left\{
\begin{array}{ll}
\po {f}(z^i_j),
    & \mbox{\rm if the literal $\ell^i_j$ is the variable, i.e., $\ell^i_j=x^i_j$},\\
\po {g}(z^i_j),
    & \mbox{\rm if $\ell^i_j$ is the negated variable, i.e., $\ell^i_j=\neg x^i_j$}.
\end{array}
\right.
\]
and meets and joins are taken in the minimal set $\set{0<1}$.

The case $F_0 =\emptyset$ can be treated similarly.
\end{proof}

\bigskip
Endowed with the tools provided by Lemma \ref{lm-pseudonegation}
we start enforcing nice lattice behavior of an algebra $\m A$ with $\typset{\m A} = \set{\tn 4}$
by associating with every join irreducible congruence $\alpha$ of $\m A$
a binary relation $\le{\alpha}$ which will turn to be a partial order on $A$ whenever
$\polsatstar{A}$ is not \npc.
If $\alpha$ is join irreducible then by $\alpha^-$ we denote its unique subcover.
Moreover pick $\set{0,1}$ to be an $(\alpha^-,\alpha)$-minimal set.
Since $\typ(\alpha^-,\alpha)=\tn 4$ we know that $\m A|_\set{0,1}$ is polynomially equivalent with the 2-element lattice, and without loss of generality we assume that $0<1$ in this lattice.
We are going to denote this choice of order on this minimal set by typing $\set{0<1}$.

Now for $a,b \in A$ put:
\[
a \le{\alpha} b \mbox{ \ iff there is a polynomial $\po f \in \poln{1}{A}$ with
$\po f \vpair{1}{0} = \vpair{b}{a}$}.
\]
Note that this relation is independent of the choice of the (ordered) $(\alpha^-,\alpha)$-minimal set $\set{0<1}$ as all $(\alpha^-,\alpha)$-minimal set are polynomially equivalent and this equivalence with $\set{0<1}$ propagates the order in the unique way.

Moreover, for further simplicity, we will use the following notation:
\begin{itemize}
  \item $a <_\alpha b$ iff $a \le{\alpha} b$ and $a \neq b$,
  \item $a \comp_\alpha b$ iff $a \le{\alpha} b$ or $b \le{\alpha} a$,
  \item $a \scomp_\alpha b$ iff $a \comp_\alpha b$ and $a \neq b$,
\end{itemize}

\begin{lm}
\label{lm-alpha-connected}
Let $\m A$ be a finite algebra from a congruence modular variety
with $\typset{\m A} \ci \set{\tn 4}$
and let $\alpha$ be a join irreducible congruence of $\m A$.
Then
\begin{enumerate}
  \item \label{lm-alpha-connected-1}
        $a\scomp_\alpha b$, whenever $\set{a,b} \in \minim{A}{\delta}{\delta'}$
        for some $\delta \prec \delta' \leq \alpha$,
  \item \label{lm-alpha-connected-2}
        for every $a \in A$ the graph $(a/\alpha, \scomp_{\alpha})$ is connected,
  \item \label{lm-alpha-connected-3}
        all unary polynomials of $\m A$ preserve the relation $\le\alpha$,
        and all polynomials of $\m A$ preserve the transitive closure of $\le\alpha$.
\end{enumerate}
\end{lm}

\begin{proof}
To see (\ref{lm-alpha-connected-1}) apply Lemma \ref{lm-malcev-chain}
to $(a,b)\in\delta'\leq\alpha=\cg{A}0 1 $ to get a chain connecting $a$ with $b$,
where each link in this chain is obtained by projecting the set $\set{0,1}$ by a unary polynomial.
Since the ends of this chain lie in the minimal set $\set{a,b}$
we can apply the unary idempotent polynomial with the range $\set{a,b}$
to put the chain into the set $\set{a,b}$.
Obviously at least one link in this chain has to be $\set{a,b}$, which finishes the proof.

To see (\ref{lm-alpha-connected-2}) we recall Lemma 2.17 of \cite{hm}
which gives that for every prime quotient $(\delta,\delta')$
each pair $(a,b)\in \delta'$ can be connected via $(\delta,\delta')$-traces and $\delta$-links.
Now, each link of the form $(c,d)\in\delta$ can be decomposed into a chain of links
modulo join irreducible congruences below $\delta$.
Thus recursively we get that any pair $(a,b)\in\alpha$ can be connected via traces with respect to the prime quotients of the form $(\beta^-,\beta)$, where $\beta$ ranges over join irreducible congruences below $\alpha$. Now, by (\ref{lm-alpha-connected-1}), each such trace is an edge in the graph $(a/\alpha, \scomp_{\alpha})$.

Finally, for (\ref{lm-alpha-connected-3}), assume that $a \le{\alpha} b$,
i.e. $\po f \vpair{1}{0} = \vpair{b}{a}$ for some $\po f \in \poln{1}{A}$.
Then obviously for $\po p \in \poln{1}{A}$ we have
$\vpair{\po p(b)}{\po p(a)}= \po p \po f\vpair{1}{0}$,
so that $\po p(a) \le{\alpha} \po p(b)$.
Now if $\po p \in \poln{s}{A}$ and $a_i \le\alpha b_i$ for all $i=1,\ldots,s$ then
\begin{eqnarray*}
\po p(a_1,a_2,a_3,\ldots,a_s)   &\le\alpha& \po p(b_1,a_2,a_3,\ldots,a_s)   \\
                                &\le\alpha& \po p(b_1,b_2,a_3,\ldots,a_s)   \\
                                &\le\alpha& \ldots                          \\
                                &\le\alpha& \po p(b_1,b_2,b_3,\ldots,b_s),
\end{eqnarray*}
so that $(\po p(a_1,a_2,\ldots,a_s), \po p(b_1,b_2,\ldots,b_s))$
lies in the transitive closure of $\le\alpha$.
\end{proof}

\bigskip

\begin{lm}
\label{lm-alpha-order}
Let $\m A$ be a finite algebra from a congruence modular variety,
$\alpha$ be a join irreducible congruence of $\m A$
and $\typset{\m A} \ci \set{\tn 4}$.
Then either $\polsatstar{A}$ is \npc \ or all the following hold:
\begin{enumerate}
  \item \label{lm-alpha-order-1}
        $\le{\alpha}$ is a partial order on $A$ without $3$-element chains,
  \item \label{lm-alpha-order-2}
        $\le{\alpha}$ is preserved by all polynomials of $\m A$,
  \item \label{lm-alpha-order-3}
        for every $a \in A$ the graph $(a/\alpha, \scomp_{\alpha})$
        is connected and acyclic,
  \item \label{lm-alpha-order-4}
        $a\scomp_\alpha b$ if and only if $\set{a,b} \in \minim{A}{\delta}{\delta'}$
        for some $\delta \prec \delta' \leq \alpha$.
\end{enumerate}
\end{lm}

\begin{proof}
To see (\ref{lm-alpha-order-1}) suppose that $a<_\alpha b <_\alpha c$,
and that this is witnessed
by unary polynomials $\po f_a$ and $\po f_c$, i.e.
$\po f_a\vpair{1}{0} = \vpair{b}{a}$ and $\po f_c\vpair{1}{0} = \vpair{c}{b}$.
Now, if $c \neq a$ refer to Lemma \ref{lm-pseudonegation}.(\ref{lm-pseudonegation-1}).
If $c=a$ and $\set{a,b}$ is a minimal set in $\m A$, Lemma \ref{lm-pseudonegation}.(\ref{lm-pseudonegation-2}) does the job.
Now suppose that $c=a$ but $\set{a,b}$ is not a minimal set in $\m A$.
Since $(a,b) \in \alpha$ then arguing as in the proof of
Lemma \ref{lm-alpha-connected}.(\ref{lm-alpha-connected-2}),
$a$ and $b$ can be connected via $(\beta^-,\beta)$-minimal sets where $\beta$
ranges over join irreducible congruences below $\alpha$.
In particular there is $d\in A$ so that $\set{a,d}$ is minimal
and therefore, by Lemma \ref{lm-alpha-connected}.(\ref{lm-alpha-connected-1}),
$\set{a,d} = \set{\po f(0), \po f(1)}$ for some $\po f \in \poln 1 A$.
Now either $\po f, \po f_a$ or $\po f, \po f_c$ put us into the setting of
Lemma \ref{lm-pseudonegation}.(\ref{lm-pseudonegation-1}).
Summing up, this shows that $\le\alpha$ is a partial order without 3 element chains.

For (\ref{lm-alpha-order-2}) use transitivity of $\leq_\alpha$ and refer to Lemma \ref{lm-alpha-connected}.(\ref{lm-alpha-connected-3}).

The last part of the Lemma would follow from (\ref{lm-alpha-order-3}) and the fact that
the binary relation $T$ defined by
\[
aTb \ \mbox{\rm iff $\set{a,b}$ is a trace with respect to some prime quotient $\delta\prec\delta'\leq\alpha$}
\]
is connected on $\alpha$-classes and is contained in $\scomp_\alpha$
(see Lemma \ref{lm-alpha-connected}).
Since $\scomp_\alpha$ is acyclic $\scomp_\alpha$ has to coincide with $T$.

\medskip
The hardest part of the Lemma is to show (\ref{lm-alpha-order-3}).
It can be inferred from Theorem 3.6 in \cite{larose:taylor}
but we decided to include our proof which seems to be more direct.

Thus suppose to the contrary that $C$ is a cycle in the bipartite graph $(A, \scomp_{\alpha})$.
Thus $\card C$ is even and $\card C \geq 4$.

Moreover let $\d_1(x,y,z),\ldots, \d_n(x,y,z), \q(x,y,z)$ be the directed Gumm terms with the properties described in Theorem \ref{thm-gumm}.
Our goal is to show that
\begin{senumerate}
\item\label{kroliczek}
$x =\d_1(x,y,z)= \ldots = \d_n(x,y,z)$ for all $x,y,z \in C$.
\end{senumerate}
Given (\ref{sect-lattices}.\ref{kroliczek}) we know that $\q(x,y,y)=x$ whenever $x,y \in C$.
This together with the last equality in Theorem \ref{thm-gumm} gives that
$\q$ is a polynomial that behaves like a Malcev operation on cycle $C$.
Now, picking $a,b \in C$ with $a <_\alpha b $
we get $b = \q(b,a,a) \leq_\alpha \q(b,b,a) =a$,
a contradiction.

\bigskip

In order to prove (\ref{sect-lattices}.\ref{kroliczek}) we will simplify notation by omitting the subscript $\alpha$ in $\le\alpha, <_\alpha, \comp_\alpha, \scomp_\alpha$.
Instead we will introduce the notation $\leq^\ell$ for $\ell \in \set{1, -1}$ where
$\leq^1$ stays for $\leq$ and $\leq^{-1}$ for $\geq$
and we will use the notation $a \scomp_C  b$ to denote that $a \scomp  b$ and $a,b \in C$.
Moreover for $a,b$ in the same $\alpha$-class we define:
\begin{itemize}
  \item $\dist a b$ to be the distance of $a$ and $b$ in the graph $(a/\alpha, \scomp)$.
\end{itemize}
If $a,b,c \in C$ we put
\begin{itemize}
  \item $\distt a b C$ to be the distance of $a$ and $b$ in the graph $(C, \scomp_C)$,
  \item $\distt a b c$ to be the length of the shortest path between $a$ and $b$ fully contained in $C$ and containing the vertex $c$,
\end{itemize}

Suppose to the contrary with (\ref{sect-lattices}.\ref{kroliczek})
that there are $a,b,c \in C$ with $\d_1(a,b,c) \neq a$.
This configuration will allow us to construct the sequence of triples $(a_i, b_i, c_i)_{i=0}^k$
of vertices in $C$,
so that after putting $d_i = \d_1(a_i,b_i,c_i)$, the following invariants will be kept:
\begin{senumerate}
  \item\label{abc-even}
        $\dist {b_i}{c_i}$ is even,
  \item\label{abc-move2}
        $\distt{b_{i+2}}{c_{i+2}}{a_{i+2}} < \distt{b_i}{c_i}{a_i}$,
  \item\label{abc-di}
        $d_i \neq a_i$,
  \item\label{abc-end}
        $a_k \in \set{b_k, c_k}$.
\end{senumerate}
The last item gives $d_k=\d_1(a_k,b_k,c_k)=a_k$ 
contrary to (\ref{sect-lattices}.\ref{abc-di}).

\medskip
First we will define the triple $(a_0,b_0,c_0)$.
If $\dist {b}{c}$ is even then we simply put $(a_0,b_0,c_0)=(a,b,c)$
so that (\ref{sect-lattices}.\ref{abc-even}) and (\ref{sect-lattices}.\ref{abc-di}) hold.
If $\dist {b}{c}$ is odd, then exactly one of the distances $\dist a b$, $\dist a c$ is even.
Suppose this is $\dist a b$.
We then move $b$ to its neighbor $b_0 \scomp_C b$ while $(a_0,c_0)$ is set to $(a,b)$.
Obviously $\dist {b_0}{c_0}$ is now even.
Moreover $b_0 <^\ell b$ for some $\ell\in\set{1,-1}$.
In particular both $b$ and $a$ are $<^\ell$-maximal.
Applying (\ref{lm-alpha-order-2}) to $b_0 <^\ell b$ we get
$\d_1(a, b_0, c) \leq^\ell \d_1(a,b,c) \neq a$.
Thus either $d_0 = \d_1(a, b_0, c) = \d_1(a,b,c) \neq a =a_0$
or $d_0 = \d_1(a, b_0, c) <^\ell \d_1(a,b,c)$ is not $<^\ell$-maximal, while $a_0=a$ is.
This shows (\ref{sect-lattices}.\ref{abc-di}).

Now, as long as $a_i \in \set{b_i, c_i}$ fails,
we create the next triple $(a_{i+1}, b_{i+1}, c_{i+1})$
by either moving $a_i$ along an edge and keeping $b_i, c_i$ untouched,
or by moving simultaneously $b_i$ and $c_i$ towards $a_i$ which stays unchanged.
More formally:
\begin{description}
  \item[\textsc{Case 1}] if $\dist{a_i}{d_i} =1$ then put
  \begin{itemize}
    \item $a_{i+1} \scomp_C a_i$ such that $\dist{a_{i+1}}{d_i} = 2$,
    \item $b_{i+1} = b_i$ and $c_{i+1} = c_i$.
  \end{itemize}
  \item[\textsc{Case 2}] if $\dist{a_i}{d_i} \geq 2$ then put
  \begin{itemize}
    \item $a_{i+1} = a_i$,
    \item $b_{i+1} \scomp_C b_i$ and $c_{i+1} \scomp_C c_i$ such that \\
            $\distt{b_{i+1}}{c_{i+1}}{a_i} \leq \distt{b_i}{c_i}{a_i}-2$.
  \end{itemize}
\end{description}
Note that the very last inequality is strong only if the initial situation is
$a_i\scomp_C b_i=c_i$ and results either in $a_i=b_{i+1} \scomp_C b_i \scomp_C c_{i+1}$
or in $a_i=c_{i+1} \scomp_C c_i \scomp_C b_{i+1}$.
Indeed, in this case $\distt{b_i}{c_i}{a_i} = \card C \geq 4$ while $\distt{b_{i+1}}{c_{i+1}}{a_i}\in\set{0,2}$.

Being in \textsc{Case 1} we know that $d_i <^\ell a_i$ for some $\ell\in\set{1,-1}$.
Then $a_{i+1}\in C$ is chosen so that $a_i >^\ell a_{i+1} \neq d_i$.
Thus we have $d_{i+1} \leq^\ell d_i$ and this inequality cannot be strong as then we would have a 3-element path $d_{i+1} <^\ell d_i <^\ell a_i$, contrary to (\ref{lm-alpha-order-1}).

Moreover, after \textsc{Case 1} is performed we fall into \textsc{Case 2}.
Thus in each two consecutive rounds the distance between $b_i$ and $c_i$ through $a_i$ decreases by at least 2, so that the invariant  (\ref{sect-lattices}.\ref{abc-move2}) is kept.
Since in \textsc{Case 2} we move both $b_i$ and $c_i$ to their neighbors,  (\ref{sect-lattices}.\ref{abc-even}) holds as well.
To see (\ref{sect-lattices}.\ref{abc-di}) note that
$b_{i+1} <^\ell b_i$ and $c_{i+1} <^\ell c_i$ for the very same $\ell \in \set{1,-1}$.
Thus $d_{i+1} \leq^\ell d_i$ while $a_{i+1} = a_i$,
which together with $\dist{a_i}{d_i} \geq 2$ gives $d_{i+1} \neq a_{i+1}$.

Finally note that (\ref{sect-lattices}.\ref{abc-move2}) gives that there is
$k \leq \distt {b_0}{c_0}{a_0}$ for which (\ref{sect-lattices}.\ref{abc-end}) holds.

This finishes the proof that $\d_1(a,b,c)= a$ whenever $a,b,c \in C$.
Now we can repeat recursively this argument for $\d_2,\ldots,\d_n$,
so that (\ref{sect-lattices}.\ref{kroliczek}), and therefore (\ref{lm-alpha-order-3}) is shown.
\end{proof}

Now we are ready to establish quite strong property enforced by tractability of $\csat{}$.

\begin{thm}
\label{si2el}
If $\m A$ is a finite subdirectly irreducible algebra from a congruence modular variety
and $\typset{\m A} =\set{\tn 4}$
then either $\card A = 2$ or $\polsatstar{A}$ is \npc.
\end{thm}

\newcommand{\Ji}{\Upsilon}

\begin{proof}
To be able to use the properties of the partial orders $\leq_\alpha$
established in Lemma \ref{lm-alpha-order} we assume that $\polsatstar{A}$ is not \npc.

Let $\mu$ be the monolith of $\m A$
and $\Ji$ be the set of all join irreducible congruences of $\m A$.
Pick a $(0,\mu)$-minimal set $N=\set{0,1}$.
Then for every $\alpha \in \Ji$ and each $(\alpha^-,\alpha)$-minimal set $U$ there is a unary polynomial $\po f_\alpha$ such that $\po f_\alpha(U) =\set{0,1}$.
Note that any unary polynomial $\po g$ for which $\po g(U)=\set{0,1}$ we have
$\po g|_U = \po f_\alpha|_U$, as otherwise we would have $0 <_\alpha 1 <_\alpha 0$, contrary to
Lemma \ref{lm-alpha-order}.
This allows us to name the elements of an $(\alpha^-,\alpha)$-minimal set by $0_\alpha, 1_\alpha$
by requiring $\po f_\alpha \vpair {1_\alpha}{0_\alpha} = \vpair 1 0$.
Note that this naming is independent of the choice of polynomials $\po f_\alpha$.
Now the ordering $0_\alpha < 1_\alpha$ determines the (bipartite) order $\leq_\alpha$ with the properties described in Lemma \ref{lm-alpha-order}.
Moreover the orderings of the form $\leq_\alpha$ are coherent in the following sense:
\begin{senumerate}
\item\label{coherence}
    for $\alpha,\beta \in \Ji$ and $(a,b)\in \alpha\cap\beta$ we have
    $a\le{\alpha} b$ iff $a\le{\beta} b$.
\end{senumerate}
Indeed suppose that $a <_\alpha b$ and $a >_\beta b$.
This gives that
$\po g_\alpha \vpair{1_\alpha}{0_\alpha} = \vpair{b}{a} = \po g_\beta \vpair{0_\beta}{1_\beta}$
for some unary polynomials $\po g_\alpha, \po g_\beta$.
On the other hand the pair $\vpair{b}{a}$ can be polynomially mapped, by say $\po f$,
onto $\set{0,1}$.
Now either $\po f \po g_\alpha\vpair{1_\alpha}{0_\alpha} = \vpair{0}{1}$
or         $\po f \po g_\beta \vpair{1_\beta}{0_\beta}= \vpair{0}{1}$, contrary to our previous choices of orders in minimal sets.

\medskip

Our first goal is to show that
\begin{senumerate}
\item\label{acyclic}
    the transitive closure $\leq$ of the sum $\bigcup_{\alpha \in \Ji} \le{\alpha}$
    is a connected partial order $A$ and it is preserved by all polynomials of $\m A$.
\end{senumerate}
The only obstacle for $\leq$ to be a partial order is the existence of a cycle of the form
\[
a_0 <_{\alpha_1} a_1 <_{\alpha_2} a_2 \ \ldots \ a_{k-1} <_{\alpha_k} a_k <_{\alpha_0} a_0.
\]
We know that $U=\set{a_0,a_k}$ is a $(\beta^-,\beta)$-minimal set for some join irreducible congruence $\beta \leq \alpha_0$.
Applying unary idempotent polynomial $\po e_U$, with the range $U$, to such a cycle we get
\[
\po e_U(a_0) \leq_{\alpha_1} \po e_U(a_1) \leq_{\alpha_2} \po e_U(a_2) \ \ldots \ \po e_U(a_{k-1}) \leq_{\alpha_k} \po e_U(a_k) \leq_{\alpha_0} \po e_U(a_0).
\]
We induct on $j=0,1,\ldots,k$ to show that $\po e_U(a_j)=a_0$.
First note that (\ref{sect-lattices}.\ref{coherence}) applied to $a_k <_{\alpha_0} a_0$
gives $a_0 \nleqslant_{\alpha_{j+1}} a_k$.
However, by the induction hypothesis
$a_0 = \po e_U(a_j) \leq_{\alpha_{j+1}} \po e_U(a_{j+1}) \in \set{a_0,a_k}$,
so that we must have $\po e_U(a_{j+1}) = a_0$, as required.
But now we have $a_0 =\po e_U(a_k)=a_k$, an obvious contradiction.

To see that the partial order $\leq$ is connected and preserved by the polynomials
simply recall the arguments used in the proof of Lemma \ref{lm-alpha-connected}.(\ref{lm-alpha-connected-2}) and Lemma \ref{lm-alpha-order}.(\ref{lm-alpha-order-2}).

\bigskip
We are working under the assumption that $\polsatstar{A}$ is not \npc.
Thus, defining the sets $F_v$ with $v\in N=\set{0,1}$ by putting
\[
F_v = \bigcap\set{\po f^{-1}(v) \ : \ \po f\in \poln 1 A \mbox{ \ and \ } \po f(A)=N},
\]
Lemma \ref{lm-pseudonegation}.(\ref{lm-pseudonegation-3})
allows us to assume that both $F_0$ and $F_1$ are nonempty.
Pick an `upper' element $u\in F_1$ and a `lower' element $d\in F_0$.
The names for them are justified by the following observation.
\begin{senumerate}
\item\label{extremal}
    For every $\po f \in \poln 1 A$ the element $\po f(u)$ is maximal in $\po f(A)$,
    while $\po f(d)$ is minimal in $\po f(A)$.
    In particular, $u$ is maximal, while $d$ is minimal in the poset $(A,\leq)$.
\end{senumerate}
Indeed, otherwise there is $w \in A$ such that
$\po f(u) <_\alpha \po f(w)$ for some $\alpha \in \Ji$.
But then $\set{\po f(u),\po f(w)}$ is a minimal set and using the consistency of our partial order
one can polynomially project, say by $\po g$,
the pair $\vpair{\po f(w)}{\po f(u)}$ onto the pair $\vpair{1}{0}$.
Composing this with a unary idempotent polynomial $\po e$ with the range $N$ we have
$\po {egf} \vpair{w}{u} = \vpair{1}{0}$, contrary to our choice of $u$.

By the very same token one shows the properties of $d$.

\medskip
We will show that in fact there are no other minimal or maximal elements in $(A,\leq)$.

\begin{senumerate}
\item\label{extremal2}
    $u$ is the largest element in the poset $(A,\leq)$, while $d$ is the smallest one.
\end{senumerate}
By symmetry of our assumptions we can restrict ourselves to show that $u$ is the largest element.
Suppose to the contrary that there is another maximal element in $(A,\leq)$.
Since the poset is connected we may assume that there are elements $b,c\in A$ with
$u > b < c$ and $c$ being maximal.

We will be using directed Gumm terms $\d_1,\ldots, \d_n, \q$,
provided by Theorem \ref{thm-gumm}, to define unary polynomials
\[
\po f_i(x) = \d_i(x,b,c), \ \  \mbox{\rm for all \ } i=1,\dots,n
\]
and show that they satisfy:
\begin{enumerate}
  \item[(i)]   $\po f_i(c) = c$,
  \item[(ii)]  $\po f_i(b) \neq c$,
  \item[(iii)] $\po f_i(u) \neq c$.
\end{enumerate}
The item (i) follows directly from the properties of the $\d_i$'s.
Moreover, for $i=1$ the item (ii) is secured by $\po f_1(b) = \d_1(b,b,c) = b <c$,
while the failure of (iii) would lead to a contradiction
\[
c = \po f_1(u) = \d_1(u,b,c) \leq \d_1(u,u,c) = u.
\]
The failure of any of the items (ii) or (iii)
at the level $(i+1)$ would give one of the following
\[
\begin{array}{ccccccc}
c &=& \po f_{i+1}(b) &=& \d_{i+1}(b,b,c) &\leq& \d_{i+1}(u,u,c),\\
c &=& \po f_{i+1}(u) &=& \d_{i+1}(u,b,c) &\leq& \d_{i+1}(u,u,c).
\end{array}
\]
In each case the maximality of $c$ yields
\[
c = \d_{i+1}(u,u,c) = \d_i(u,c,c) \geq \d_i(u,b,c) = \po f_i(u),
\]
and now the maximality of $\po f_i(u)$ gives $\po f_i(u)=c$,
a contradiction with the induction hypothesis.

\medskip
Now,
\[
c = \q(b,b,c) \leq \q(u,c,c) = \d_n(u,c,c),
\]
together with maximality of $c$ gives
\[
c = \d_n(u,c,c) \geq \d_n(u,b,c) = \po f_n(u),
\]
so that maximality of $\po f_n(u)$ gives $\po f_n(u)=c$,
contrary to (iii).
This contradiction shows (\ref{sect-lattices}.\ref{extremal2}).

\bigskip
To conclude our proof we will strengthen (\ref{sect-lattices}.\ref{extremal2}) to:
\begin{senumerate}
\item\label{extremal3}
    $A=\set{u,d}$.
\end{senumerate}
First suppose that $\set{0',1'}$ is a $(0,\mu)$-minimal set.
Thus for every $\alpha\in\Ji$ we have $d \leq 0' <_{\alpha} 1' \leq u$.
Now if $d < 0'$ or $1' < u$ then we can pick an element $a$ such that either
$a <_\alpha 0'$ or $1' <_\alpha a$ for some $\alpha\in\Ji$.
In any case we will have a 3-element directed path in the poset $(A, <_\alpha)$
which is not possible in view of Lemma \ref{lm-alpha-order}.(\ref{lm-alpha-order-1}).
Thus $\set{d,u}$ is the only $(0,\mu)$-minimal set of $\m A$.

Now, suppose there are elements $a,b \in A$ such that
$d \leq a <_\alpha b \leq u$.
Then $\set{a,b}$ is a minimal set
which has to be the range of some unary idempotent polynomial $\po e$ of $\m A$.
But then
$\po e(d) \leq \po e(a) < \po e(b) \leq \po e(u)$
and consequently monotonicity of $\po e$ gives
$\po e \vpair {d}{u} = \vpair {a}{b}$.
In particular $(a,b)\in\mu$ and in fact the set $\set{a,b}$ is $(0,\mu)$-minimal,
so that $\set{a,b} = \set{d,u}$.
This shows that $A$ can not have any other elements than $d$ or $u$.
\end{proof}

Directly from Theorem \ref{si2el} we get

\begin{cor}
\label{2el}
Let $\m A$ be a finite algebra from a congruence modular variety
and $\typset{\m A} =\set{\tn 4}$.
Then either $\m A$ is a subdirect product of $2$-element algebras each of which is polynomially equivalent to the $2$-element lattice,
or $\m A$ has a subdirectly irreducible homomorphic image $\m A^\prime$ such that
$\polsatstar{A^\prime}$ is \npc.
\myqed
\end{cor}

\bigskip
The next example shows that a subdirect product of $2$-element algebras each of which is polynomially equivalent to the $2$-element lattice need not be polynomially equivalent to a distributive lattice.

\begin{ex}
Let $\m A=(A,\po m)$ be a subreduct of $(\set{0,1},\meet, \join)^3$, with
\begin{itemize}
  \item $A=\set{(1,1,1), (0,1,1), (1,0,1), (1,1,0)}$
  \item and $\po m$ being the majority operation
        $\po m(x,y,z) = (x \join y) \meet (y \join z) \meet (z \join x)$.
\end{itemize}
Then
\begin{itemize}
  \item $\m A$ belongs to congruence distributive variety
  \item $\m A$ is a subdirect product of algebras polynomially equivalent to two element lattices,
  \item $\m A$ is not polynomially equivalent to a distributive lattice.
\end{itemize}
\end{ex}

\begin{proof}
The first two items are obvious.
To see the third one note that, up to isomorphism,
there are only two four element lattices:
\begin{itemize}
  \item the four element chain,
  \item the four element Boolean lattice.
\end{itemize}
On the other hand, for three pairwise different elements $a,b,c \in A$
we have $\po m(a,b,c) = \o 1$, where $\o 1 = (1,1,1)$.
Sending isomorphically, say by $h$, all possible 3-element tuples from $A$
into one of the above 4-element lattices
we simply cannot find a room for $h(\o 1)$
under the assumption that $\po m$ preserves lattice order. 
\end{proof}

\ssection{Polynomial time algorithms
\label{sect-algo}}

\newcommand{\J}[2]{J_{#1,#2}}
\newcommand{\ess}[1]{\textsf{Ess}\left( #1 \right)}

Combining Corollaries \ref{cor-decomposition}, \ref{cor-solvable} and \ref{2el}
we get the following Theorem.

\begin{thm}
\label{thm-nil-dl-decomp}
Let $\m A$ be a finite algebra from a congruence modular variety such that
$\polsatstar{A^\prime}$ is not \npc \ for every quotient $\m A^\prime$ of $\m A$.
Then $\m A$ is isomorphic to a direct product $\m N \times \m D$,
where $\m N$ is a nilpotent algebra
and $\m D$ is a subdirect product of $2$-element algebras each of which is polynomially equivalent to the $2$-element lattice.
\myqed
\end{thm}

\bigskip
The aim of this section is to prove a partial converse to Theorem \ref{thm-nil-dl-decomp}
where nilpotency is strengthened to supernilpotency.

\begin{thm}
\label{thm-supernil-dl-decomp}
Let $\m A$ be a finite algebra from a congruence modular variety
that decomposes into a direct product $\m N \times \m D$,
where $\m N$ is a supernilpotent algebra
and $\m D$ is a subdirect product of $2$-element algebras each of which is polynomially equivalent to the $2$-element lattice.
Then for every quotient $\m A^\prime$ of $\m A$ the problem
$\polsatstar{A^\prime}$ is solvable in polynomial time.
\end{thm}

Before proving this Theorem note that
if an algebra $\m A$ decomposes into a direct product of the form described above
then all its quotients decompose in the very same way.
This is an immediate consequence of the fact that the product of the form
$\m N \times \m D$ has no skew congruences
which in turn follows from $\typset{\m D} \ci \set{\tn 4}$.

The proof of Theorem \ref{thm-supernil-dl-decomp} splits into two parts.
We show that for both factors of $\m A$ the problem has polynomial time solution.
Actually we will show that in both cases if the polynomial equation
$\po t(x_1,\ldots,x_n)=\po s(x_1,\ldots,x_n)$
has a solution in $A^n$ then it has a solution in a relatively small subset $S$ of $A^n$,
namely in a subset of size bounded by a polynomial in $n$.
The reader should be however warned here that we are not going to show that all solutions are contained in this small set $S$.

\begin{thm}
\label{thm-dl}
Let $\m D$ be a subdirect product of finitely many $2$-element algebras
each of which is polynomially equivalent to the $2$-element lattice.
Then $\polsatstar{D}$ is solvable in polynomial time.
\end{thm}

\begin{proof}
The basic observation is that for the $2$-element lattice $\m D$,
and therefore for every algebra polynomially equivalent to the $2$-element lattice,
the problem $\polsatstar{D}$ is solvable in polynomial time by a very special algorithm.

Indeed, if $\po t, \po s \in \pol D$ the equation
$\po t(\o x) = \po s(\o x)$ has a solution, say $(a_1,\dots,a_n)$,
then both $\po t(a_1,\dots,a_n)$ and $\po s(a_1,\dots,a_n)$ have the same value $a$.
But for a polynomial $\po t$ over the $2$-element lattice
one can easily show,
that if $\po t(a_1,\dots,a_n)= a$ then $\po t(a,\dots,a)= a$.
Indeed, by the monotonicity of the polynomials of $\m D$ we have
\[
\po t(0,\ldots, 0) \leq \po t(a_1,\dots,a_n) \leq \po t(1,\ldots, 1)
\]
and if $\po t(a_1,\dots,a_n)=0$ then $\po t(0,\ldots, 0)$ has to be $0$ as well.
Similarly $\po t(a_1,\dots,a_n)=1$ implies $\po t(1,\ldots, 1)=1$.

Therefore, to determine if $\po t(\o x) = \po s(\o x)$
has a solution over $\m D$ it suffices to show whether
$\po t(a,\dots,a)= \po s(a,\ldots,a)$ for some $a\in D$.

\medskip
We say that an algebra $\m A$ has \usp, or \USP \ for short,
if for every polynomial $\po t(\o x)\in \pol A$ and $a \in A$
\[
\big( \exists \o x \ \ \po t(x_1,\ldots,x_n) =a \big) \Rightarrow \po t(a,\dots,a)= a
\]
What we have just shown is that the $2$-element lattice has \USP,
and that $\polsatstar{A}$ is polynomially time solvable for every finite algebra $\m A$ with \USP.

Now we can conclude the proof by noting that a subdirect product of algebras with \USP,
has \USP \ itself.
Actually \USP \ is preserved under forming homomorphic images, subalgebras, products or reducts.
\end{proof}

\bigskip

The reduction of searching a solution of an equation in supernilpotent realm to a relatively small set is much more involved than in lattice case.
Our proof is modeled after the Ramsey type argument
introduced by Mikael Goldmann and Alexander Russell in \cite{goldman-russell}
for nilpotent groups, and later cleaned up by G\'{a}bor Horv\'{a}th \cite{horvath:positive}
in the realm of nilpotent groups and nilpotent rings.

\begin{thm}
\label{thm-supernil}
Let $\m A$ be a finite supernilpotent algebra from a congruence modular variety.
Then $\polsatstar{A}$ is solvable in polynomial time.
\end{thm}

\begin{proof}
Suppose now that $\m A$ is a finite nilpotent algebra from a congruence modular variety.
Then $\m A$ generates a variety in which every algebra is nilpotent, and therefore congruence permutable. In particular $\m A$ has a Malcev term $\po d(x,y,z)$.

From now on we fix an element $0$ of $\m A$ and define a binary operation $+$ by putting:
\[
x + y = \po d(x,0,y).
\]
Unfortunately the binary operation $+$ does not need to be associative. Thus in longer sums we adopt the convention of associating to the left.
More formally, if $\tup{a_1,a_2,\ldots,a_\ell}$ is an ordered list of elements of $\m A$
then by $\sum \tup{a_1,a_2,\ldots,a_\ell}$ we mean $((a_1+a_2)+a_3)+\ldots+a_\ell$.

According to Corollary 7.4 in \cite{fm}, we know that for every $b,c \in A$
the function $x \mapsto \po d(x,b,c)$ is a permutation of $A$.
In particular
\[
\po d(x,y,0) = 0 \mbox{ \ \ iff \ \ } x=y.
\]
Thus, each equation $\po t(\o x) = \po s(\o x)$ can be equivalently replaced by an equation
of the form $\po w(\o x) = 0$,
where $\po w(\o x) = \po d(\po t(\o x), \po s(\o x), 0)$
has linear size in terms of the size of the original equation.

Our polynomial time algorithm for checking whether $\po w(\o x) =0$ has a solution is based on the following phenomena of supernilpotent algebras:
\begin{senumerate}
\item\label{d-non-zero}
For every finite supernilpotent algebra $\m A$ there is a positive integer $d$ such that
every equation of the form $\po w(\o x)=0$ has a solution iff
it has a solution with at most $d$ non-zero values for the $x_i$'s,
i.e. $\cardd \equa{\o x \neq 0} \leq d$.
\end{senumerate}

Given (\ref{sect-algo}.\ref{d-non-zero}) we simply check
if $\po w(x_1,\ldots,x_n)=0$ has a solution among $ {n \choose d} \cdot \card{A}^d$ possible evaluations of the $x_i$'s with $\cardd \equa{\o x \neq 0} \leq d$.
Unfortunately the degree $d$ of the polynomial bounding the run time of the algorithm can be really huge, as it is obtained by a Ramsey type argument applied to the numbers:
\begin{itemize}
  \item $k$ -- the degree of supernilpotency of the algebra $\m A$,
  \item $C = \card{A}^{k \cdot \card{A}}$,
  \item $m = (k-1)! \cdot \card{A}$
\end{itemize}
to get that:
\begin{senumerate}
\item\label{ramsey}
There is a positive integer $d$ such that
for every set $S$ with $\card{S} \geq d$
and every coloring of all at most $(k-1)$-element subsets of $S$ with $C$ colors
there exists $m$-element subset $T$ of $S$
such that all at most $(k-1)$-element subsets of $T$ with the same number of elements
have the same color.
\end{senumerate}
For a proof of the above statement we refer e.g.
to Theorem 2, Chapter 1 in the monograph \cite{grs}.

\medskip
Now, to prove (\ref{sect-algo}.\ref{d-non-zero}) we will show that
\begin{senumerate}
\item\label{d-minus-m}
Each solution $\o b = b_1,\ldots,b_n$ of $\po w(x_1,\ldots,x_n)=0$
with $\cardd \equa {\o b \neq 0} > d$
can be replaced with a solution $\o {b'} = b'_1,\ldots,b'_n$
with $\cardd \equa {\o {b'} \neq 0} = \cardd \equa {\o b \neq 0}-m$.
\end{senumerate}

\medskip
\noindent
For the rest of the proof we fix a solution $\o b = b_1,\ldots,b_n$ of $\po w(x_1,\ldots,x_n)=0$.
Now, a very careful reading of Chapter XIV of \cite{fm}, especially the proof of Lemma 14.6,
allows us to represent $\po w(\o x)$ as
$\sum \tup{\po w_1(\o x),\ldots, \po w_k(\o x)}$,
where $k$ is the degree of supernilpotency, and therefore nilpotency, of $\m A$
and each $\po w_\ell(\o x)$ has the form
\[
\po w_\ell(\o x) = \sum \tup{c_\ell, \po t_{\ell,1}(\o x),\ldots, \po t_{\ell,n_\ell}(\o x)},
\]
with
\begin{itemize}
  \item $c_\ell \in A$,
  \item $\po t_{\ell,j}(\o x)=0$
        whenever $x_i =0$ for at least one $i \in \ess{\po t_{\ell,j}}$,
        where $\ess{\po t_{\ell,j}}$ is the set of numbers of variables on which
        $\po t_{\ell,j}$ essentially depends,
  \item for each $a\in A$ the sublist $\J \ell a$ of $\tup{1,\ldots,n_\ell}$
        consisting of the $j$'s for which $\po t_{\ell,j}(\o b)=a$
        is convex in $\tup{1,\ldots,n_\ell}$.
\end{itemize}
The second item above simply means that the $\po t_{\ell,j}$'s are commutator expressions
and therefore our assumption that $\m A$ is $k$-supernilpotent gives
\begin{itemize}
  \item $\card{\ess{\po t_{\ell,j}}} < k$.
\end{itemize}

Now, suppose that the set $S = \equa{\o b \neq 0}$ is too big, i.e. $\cardd{S} > d$.
Define a coloring $\fj$ of (at most $(k-1)$-element) subsets of $S$ by $C$ colors as follows.
For a subset $I \ci S$ its color $\fj_I$ is set to be a function of the form
$\set{1,\ldots,k} \times A \map \set{0,1,\ldots, \card{A}-1}$
determined by
\[
\fj_I(\ell,a) = \cardd\set{j\in \J \ell a : I\ci \ess{\po t_{\ell,j}}} \mod \card{A}.
\]
Now, (\ref{sect-algo}.\ref{ramsey}) supplies us with $T \ci S$ such that $\card{T}=m$
and $\fj_{I_1} = \fj_{I_2}$ whenever $I_1,I_2 \ci T$ and $\card{I_1}=\card{I_2} <k$.

We modify $\o b$ to $\o {b'}$ by zeroing the $x_i$'s with $i\in T$, i.e.
\[
b'_i =
\left\{
\begin{array}{ll}
0,      &\mbox{\rm if \ } i \in T,\\
b_i,    &\mbox{\rm otherwise.}
\end{array}
\right.
\]
Obviously $\cardd \equa {\o {b'} \neq 0} = \cardd \equa {\o b \neq 0}-m$,
as required in (\ref{sect-algo}.\ref{d-minus-m}).
To prove that $\po w(b'_1,\ldots,b'_n) = 0$
we will show that
\begin{senumerate}
\item\label{b-to-bprime}
for each $\ell=1,\dots,k$ and $a\in A$ we have
\[
\sum \tup{c,\ldots, \po t_{\ell,j}(\o {b'}),\ldots }_{j \in \J \ell a} =
\sum \tup{c,\ldots, \po t_{\ell,j}(\o {b}),\ldots }_{j \in \J \ell a}.
\]
\end{senumerate}
Note that in the sum on the left hand side some of the summands switched from $a$ to $0$
(if $a \neq 0$).
Let $Z$ collects the numbers $j$ of such summands.
Obviously
$Z = \set{j\in \J \ell a : T \cap \ess{\po t_{\ell,j}} \neq \emptyset}$.
We will show that
\begin{senumerate}
\item\label{divisible}
$\card{Z}$ is divisible by $\card{A}$.
\end{senumerate}
Given (\ref{sect-algo}.\ref{divisible})
we argue towards (\ref{sect-algo}.\ref{b-to-bprime}) as follows.
We have already noticed that
$x \mapsto x+a = \po d(x,0,a)$ is a permutation of $A$.
Let $\sigma$ be the order of this permutation, so that
\[
\sum \tup{x,\underbrace{a,\ldots,a}_{\text{$\sigma$ times}}}=x.
\]
Moreover the fact that $x+0 = \po d(x,0,0)=x$ allows us to omit all the $0$'s in the lists under the sums. This gives the first equality in the display below.
\begin{eqnarray*}
\sum \tup{c,\ldots, \po t_{\ell,j}(\o {b'}),\ldots }_{j \in \J \ell a}
&=&
\sum \tup{c,\underbrace{a,\ldots\ldots\ldots,a}_{\text{$(\card{\J \ell a}-\card{Z})$ times}}}
\\
&=&
\sum \tup{c,\underbrace{a,\ldots\ldots,a}_{\text{$\card{\J \ell a}$ times}}}
\\
&=&
\sum \tup{c,\ldots, \po t_{\ell,j}(\o {b}),\ldots }_{j \in \J \ell a}
\end{eqnarray*}
The second equality in this display uses the fact that $\sigma$ divides $\card{Z}$, which follows from (\ref{sect-algo}.\ref{divisible})
and Lemma 14.7 in \cite{fm} telling us that $\sigma$ dives $\card{A}$.

\medskip

Thus, we are left with the proof of (\ref{sect-algo}.\ref{divisible}).
To do this, for $I \ci T$ define
\[
Z_I = \set{ j \in \J \ell a : I \ci \ess{\po t_{\ell,j}}}
\]
and observe that
\[
Z = \bigcup_{i \in T} Z_\set{i}
\ \ \ \mbox{\rm and } \ \ \
Z_I = \bigcap_{i \in I} Z_\set{i}.
\]
Thus the inclusion-exclusion principle, together with $\card{\ess{\po t_{\ell,j}}} < k$, gives
\[
\card{Z} = \sum_{I \ci T \atop 0<\card{I}<k} \left(-1\right)^{\card{I}+1} \cdot \card{Z_I}.
\]
However we know that $\card{Z_I}$ is (modulo $\card{A}$) nothing else but $\fj_I(\ell,a)$.
Since $\fj_I$ depends only on the size of $I$ we put
$\zeta_q = \fj_I(\ell,a)$ for $\card{I}=q$ to get that modulo $\card{A}$ we have
\[
\card{Z} =  \sum_{q=1}^{k-1} (-1)^{q+1}{m \choose q} \cdot \zeta_q.
\]
To conclude the proof that $\card{Z}$ is divisible by $\card{A}$
observe that for $q=1,\ldots,k-1$
all the binomial coefficients $m \choose q$ are divisible by $\card{A}$, as
$m$ was set to be $(k-1)! \cdot \card{A}$.
\end{proof}

\ssection{Simultaneous satisfiability of many circuits}
\label{sect-multi-csat}

This section is devoted to the problem $\mcsat{}$.
All we have to do is to prove Corollary \ref{int-mcsat}.

\medskip
\noindent
{\bf Corollary \ref{int-mcsat}.}
{\em
Let $\m A$ be a finite algebra from a congruence modular variety.
\begin{enumerate}
\item
\label{int-csat-nonpc}
    If $\m A$ has no quotient $\m A^\prime$ with $\mcsat{A^\prime}$ being \npc
    then $\m A$ is isomorphic to a direct product $\m M \times \m D$,
    where $\m M$ is an affine algebra
    and $\m D$ is a subdirect product of $2$-element algebras
    each of which is polynomially equivalent to the $2$-element lattice.
\item
\label{int-csat-ptime}
    If $\m A$ decomposes into a direct product $\m M \times \m D$,
    where $\m M$ is an affine algebra
    and $\m D$ is a subdirect product of $2$-element algebras each of which is polynomially equivalent to the $2$-element lattice,
    then for every quotient $\m A^\prime$ of $\m A$ the problem
    $\csat{A^\prime}$ is solvable in polynomial time.
\end{enumerate}
}

\medskip
\begin{proof}
First note that every instance of $\csat{A}$ is also a instance of $\mcsat{A}$,
so that $\mcsat{A}$ is \npc whenever $\csat{A}$ is \npc.
Moreover, $\mcsat{A}$ can be treated as a problem of satisfiability
of systems of equations of the form
\[
\po g_1(x_1,\ldots,x_n)=\po g_2(x_1,\ldots,x_n)=\ldots=\po g_k(x_1,\ldots,x_n).
\]
Thus $\mcsat{A}$ is in fact a subproblem of $\scsat{A}$,
so that Gaussian elimination type algorithms for affine $\m A$'s show that in this case
$\mcsat{A} \in \ptime$ (see also Theorem \ref{int-scsat-zadori}).
Therefore, by Theorem \ref{int-csat} we are left with the following two classes of algebras:
\begin{itemize}
 \item nilpotent non-affine algebras,
 \item subdirect products of algebras each of which is polynomially equivalent to the $2$-element lattice.
\end{itemize}

In case $\m A$ is nilpotent we can easily interpret $\scsat{A}$ into $\mcsat{A}$.
Indeed, nilpotent $\m A$ has a Malcev polynomial $\po d$, such that for all $a,b \in A$
the function $x\mapsto\po d(x,a,b)$ is a permutation.
Thus, after arbitrarily fixing $a \in A$, the system of equations over $\m A$
\begin{eqnarray*}
        \po g_1(x_1,\ldots,x_n) &=& \po h_1(x_1,\ldots,x_n)\\
                                &\vdots&\\
        \po g_k(x_1,\ldots,x_n) &=& \po h_k(x_1,\ldots,x_n),
\end{eqnarray*}
can be equivalently rewritten to the following instance of $\mcsat{A}$
\[
\begin{array}{ccccccc}
\po d(\po g_1(x_1,\ldots,x_n),\po h_1(x_1,\ldots,x_n),a)
&=&\ldots&=&
\po d(\po g_k(x_1,\ldots,x_n),\po h_k(x_1,\ldots,x_n),a)
&=&
a.
\end{array}
\]
This interpretation, together with Theorem \ref{int-scsat-zadori}, makes $\mcsat{A}$ \npc
whenever $\m A$ is nilpotent but not affine.

To see that subdirect products of algebras each of which is polynomially equivalent to the $2$-element lattice stay on the polynomial side recall, from the proof of Theorem \ref{thm-dl},
that such algebras have \usp.
Thus checking if
\[
\begin{array}{ccccccc}
\po g_1(x_1,\ldots,x_n)
&=&
\po g_2(x_1,\ldots,x_n)
&=& \ldots &=&
\po g_k(x_1,\ldots,x_n)
\end{array}
\]
has a solution in $\m A$, reduces to finding $a\in A$  with
\[
\begin{array}{ccccccc}
\po g_1(a,\ldots,a)
&=&
\po g_2(a,\ldots,a)
&=& \ldots &=&
\po g_k(a,\ldots,a).
\end{array}
\]
\end{proof}

\ssection{Equivalence of Circuits}
\label{sect-ceqv}

This section considers the equivalence of circuits as defined in the problem $\ceqv{}$.
Our results in this direction are covered by Theorem \ref{int-ceqv-nonnil}.

\medskip
\noindent
{\bf Theorem \ref{int-ceqv-nonnil}.}
{\em
Let $\m A$ be a finite algebra from a congruence modular variety.
If $\m A$ has no quotient $\m A^\prime$ with $\ceqv{A^\prime}$ being \conpc
then $\m A$ is nilpotent.
}

\medskip
\begin{proof}
First note that if $\tn 3$ or $\tn 4$ is in $\typset{\m A}$ then
$\m A$ has a minimal set $U=\set{0,1}$ such that $\m A|_U$ is polynomially equivalent to
either $2$-element Boolean algebra or $2$-element lattice.
But $\ceqv{}$ is \conpc for both of these small algebras (see Example \ref{ex-lattices}).
Arguing like in the proof of Theorem \ref{thm-type-3} this intractability can be carried up to $\ceqv{A}$.

Thus we are left with solvable algebras, i.e. with $\typset{\m A} \ci \set{\tn 2}$.
To force such algebra $\m A$ to be nilpotent
we can argue to the contrary like in the proof of Corollary \ref{cor-solvable}
to produce its subdirectly irreducible quotient monolith of which does not centralize $1$.
This allows us to mimic the proof of Lemma \ref{lm-solvable}.
Indeed the polynomials $\po t_G(\o x)$ and $\po t_\Phi(\o x)$ produced there
(to encode graph colorability or Boolean satisfiability, respectively)
take only two values: $a$ and $e$.
Now, instead of considering the satisfiability of the equation $\po t(\o x) = a$
we check whether $\po t(\o x)$ always return the value $e$.
\end{proof}


\end{document}